\documentclass[journal]{IEEEtran}
\usepackage[square,numbers,sort&compress]{natbib}
\usepackage{url}
\usepackage{physics} 
\usepackage{graphicx}    
\usepackage{amsmath}
\usepackage{mathtools}
\usepackage[T1]{fontenc}
\usepackage[utf8]{inputenc}
\usepackage[english]{babel}
\usepackage{amssymb,setspace,amsthm,calrsfs}
\usepackage{setspace}
\usepackage{pgfplotstable,pstool,psfrag}
\usepackage{pgfplots}
\pgfplotsset{compat=newest,compat/show suggested version=false}
\usepackage{float}
\usepackage{tabularx,colortbl}
\usepackage{ltablex}
\usepackage{booktabs}
\usepackage{algorithm2e}
\usepackage{picins}
\usepackage{tabularx}
\usepackage{supertabular}
\usepackage{stmaryrd}
\usepackage{scalerel}
\usepackage{tikz-cd}
\usepackage{color}
\usepackage{xcolor}

\newcommand{\R}{\mathbb{R}}
\newcommand{\N}{\mathbb{N}}

\newtheorem{thm}{Theorem}
\newtheorem{cor}{Corollary}

\newtheorem{dfn}{Definition}

\newtheorem{rem}{Remark}

\title{\LARGE \bf
Predictability and Fairness in Social Sensing}
\author{Ramen Ghosh, Jakub Mare\v{c}ek,~\IEEEmembership{Member,~IEEE,} Wynita M. Griggs,~\IEEEmembership{Member,~IEEE,} Matheus Souza and Robert N. Shorten,~\IEEEmembership{Senior Member,~IEEE}%
\thanks{This work was supported in part by the Science Foundation Ireland under Grant 16/IA/4610 and was co-funded under the European Regional Development Fund through the Southern \& Eastern Regional Operational Programme to Lero - the Irish Software Research Centre.  The work was also partially supported by Funda\c{c}\~ao de Amparo \`a Pesquisa do Estado de S\~ao Paulo (FAPESP/Brazil) under Grant 2016/19504-7.
Jakub has been supported by the OP RDE funded project CZ.02.1.01/0.0/0.0/16\_019/0000765 ``Research Center for Informatics''.
Parts of this work appeared in \cite{souza_conference}.
}
\thanks{R. Ghosh is with the School of Electrical and Electronic Engineering, University College Dublin, Ireland.
{\tt\small ramen.ghosh@ucdconnect.ie}}
\thanks{J. Marecek is with the Czech Technical University in Prague, the Czech Republic.
{\tt\small jakub.marecek@fel.cvut.cz}}
\thanks{W. M. Griggs is with the Department of Civil Engineering and the Department of Electrical and Computer Systems Engineering, Monash University, Clayton, Victoria, 3800, Australia.
{\tt\small wynita.griggs@monash.edu}}
\thanks{M. Souza is with the School of Electrical and Computer Engineering, University of Campinas, Brazil.
{\tt\small msouza@fee.unicamp.br}}
\thanks{R. N. Shorten is with the the Dyson School of Design Engineering, Imperial College London, South Kingston, UK.
{\tt\small r.shorten@imperial.ac.uk}}
}
\begin{document}
\maketitle
\thispagestyle{empty}
\pagestyle{empty}
\begin{abstract}
We consider the design of distributed algorithms that govern the manner in which agents contribute to a social sensing platform. Specifically, we are interested in situations where fairness among the agents contributing to the platform is needed. A notable example are platforms operated by public bodies,
where fairness is a legal requirement. The design of such distributed systems is challenging due to the fact that we wish to simultaneously realise an efficient social sensing platform, but also deliver a predefined quality of service to the agents (for example, a fair opportunity to contribute to the platform). In this paper, we introduce \textit{iterated function systems} (IFS) as a tool for the design and analysis of systems of this kind. We show how the IFS framework can be used to realise systems that deliver a predictable quality of service to agents, can be used to underpin contracts governing the interaction of agents with the social sensing platform, and  which are efficient. 
To illustrate our design via a use case, we consider a large, high-density network of participating parked vehicles. When awoken by an administrative centre, this network proceeds to search for moving missing entities of interest using RFID-based techniques. We regulate which vehicles are actively searching for the moving entity of interest at any point in time. In doing so, we seek to equalise  vehicular energy consumption across the network. This is illustrated through simulations of a search for a missing Alzheimer's patient in Melbourne, Australia. Experimental results are presented to illustrate the efficacy of our system and the predictability of access of agents to the platform independent of initial conditions.
\end{abstract}

\begin{IEEEkeywords}
Smart Cities, Internet of Things (IoT), Social sensing, Radio Frequency Identification Systems, Ergodicity, Control theory.
\end{IEEEkeywords}

\section{INTRODUCTION}
In many applications, a physical phenomenon can be sensed by collecting data collaboratively \cite{howe2006rise,Luetal2014,6740844,artikis2014heterogeneous,Cogilletal2014,6517107,7365472,7397993,7932851,8382177,9105079}, either from humans directly, or from devices acting on their behalf. This is variously known as (spatial) crowdsourcing \cite{howe2006rise,To2014,8316812}, (mobile) crowd sensing \cite{8429062,guo2015mobile}, or social sensing \cite{8666667,wang2015social}. Often, there are more humans and their devices that could contribute to the platform than the platform can utilize at any given time. 
In such situations, it may be desirable that the algorithms that govern the sensing process have certain properties such as fair and predictable distribution of the work among agents. These requirements are becoming very important and arise in many situations:
\begin{itemize}
    \item For example, in applications where the platform is operated by a branch of a government, which often has a legal requirement to ensure  fair and equitable treatment of citizens.
    In many such situations, the rights of sub-groups and their  representation in the sensing platform must be considered in the mechanism-design process.
One may think of this as akin to ensuring the ability of citizens to vote in a voting system. 
\item Another example arises in applications where there are certain types of incentives on offer. In such situations, one is interested in ensuring that agents have an equal chance to avail of these incentives. Often, this is a legal requirement (e.g., stemming from lottery regulations). 
\item A further example occurs in applications where written contracts between the participants and the platform are issued, which should involve guarantees of fairness of predictability as some quality-of-service measures. 
\item Finally, the same requirements arise in situations where fair and predictable access is mandated for legal or other reasons. Among others, the European Commission \cite{Euractive2021} aims to regulate certain ``high-risk'' AI applications. For example, when participants report pollution levels in a neighbourhood and this information is then used to route vehicular traffic, a sensing platform may be legally required to provide a fair and predictable access to participants from all neighbourhoods, to make sure that certain neighbourhoods do not see excessive traffic due to their under-representation in the pollution-sensing scheme.
\end{itemize}



Generally speaking, fairness issues have not yet been widely considered in the context of the design of social sensing systems. Typically, 
in social sensing systems, information and actuation capabilities are crowd-sourced to generate functionality to control and influence ensemble behaviour, 
with the primary objective often being the efficiency of the platform, frequently with some privacy guarantees. Examples of such situations in smart-city applications
include sensing to detect and allocate parking spaces, electric charge points, or as we have mentioned, ambient pollution
in cities. While prior papers deal with many aspects of crowd-sensing problems, most have focused on the design
of efficient crowd-sensing systems. Efficient could mean, for example, systems that minimize energy consumption, or
have the smallest pollution footprint.

Our interest in this paper is somewhat different and stems from a desire to
develop systems that are not only efficient, but in which agents' rights to contribute to the platform are characterised by certain fairness and predictability
constraints (perhaps out of economic or legal considerations). As we have mentioned, in many such situations, the rights of
sub-groups and the representation in a sensing platform must be coded as part of the algorithmic design process. 

Often, in such systems, a feedback signal is also used to regulate  the number of participants contributing to the social-sensing platform at any given time. For example, in many situations, a generalisation of a price signal could be used to encourage/discourage agents to contribute to a crowd-sensing platform. In other situations, we might wish to regulate the number of participants contributing to a task to minimize energy consumption. The architecture of such systems is depicted in Figure \ref{fig:illustration1}. 

\begin{figure}[bt]
\centering
\includegraphics[width=0.92\columnwidth]{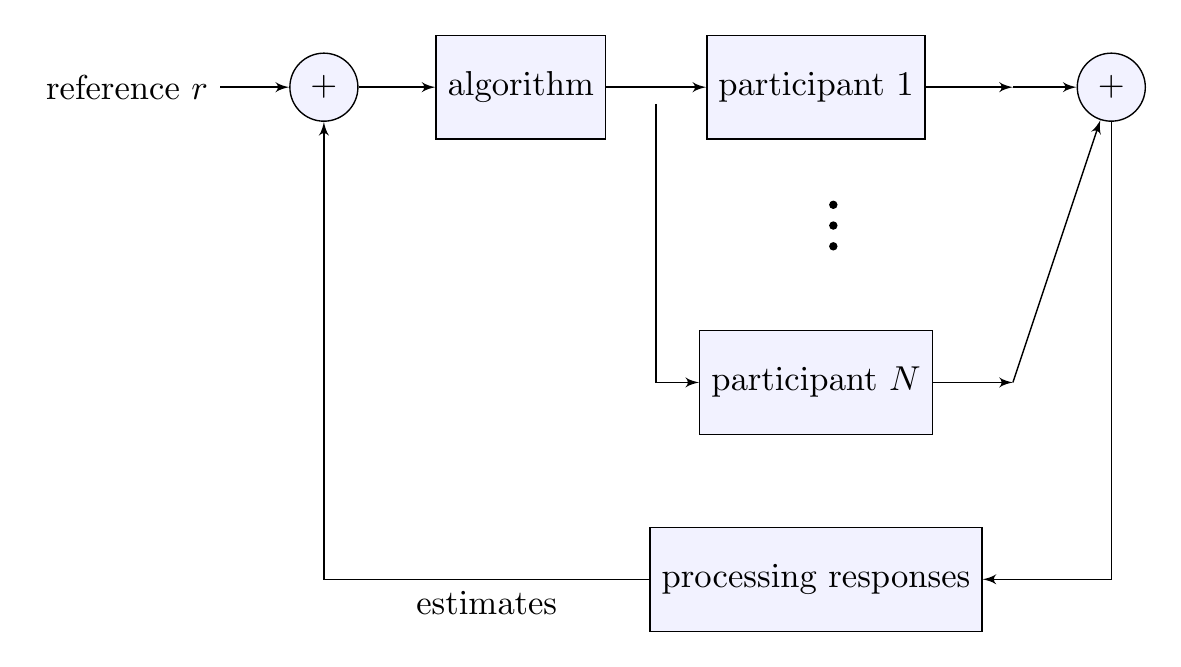}
\caption{Basic setting. There are more potential participants than  a crowdsensing platform will need to use on a given day. The platform selects the participants, processes their answers, and depending on some reference signal, may invite further participants.}\label{fig:illustration1}
\end{figure}

Designing sensing systems of this form is challenging. Clearly, we wish to allocate access to the sensing platform in a manner that is not wasteful, which gives an optimal return on the use of the resource for society, and which, in addition, gives a guaranteed level of service to each of the agents competing for that resource. Roughly speaking, when we design such systems, we seek to meet the following objectives. 
\begin{itemize}
\item Our first objective is to solve the regulation problem. For example, we may wish to ensure that a certain number of agents contribute to the sensing platform at any time (for example, to minimize the monetary burden on the sensing system provider, or to minimize the utilization of some shared communication links). 
\item Second, we would then like to develop sensing systems with the optimal behaviour. In the above example, we might select agents that have a lower energy footprint in measuring sensing data.
\item The third objective focuses on the effects of the control on the microscopic properties of the agent population. In particular, we may wish that each agent, on average, receives a fair opportunity to contribute data to the platform, or at a much more fundamental level, we may wish the average allocation access to the platform for each agent over time is a stable quantity that is entirely predictable, and which does not depend on the initial conditions. The need for fair access to the resource may arise for a number of reasons. For example, agents may have paid to write to the platform, or may even be mandated as part of some legal requirement. \cite{santos2005accuracy}
\end{itemize}

The first two of the above objectives are classical control theoretic objectives. The third is somewhat new, even in the context of control engineering. In what follows, we shall show that all three objectives can be met
in the design of our crowd-sourcing algorithms. To do this, our principal tool will be to develop techniques, whereby we establish conditions that guarantee ergodicity. Specifically, by ergodicity we mean the existence of a unique invariant measure, to which the system is attracted in a statistical sense, irrespective of the initial conditions. Thus, the design of systems for deployment in multi-agent applications must consider not only the traditional notions of regulation and optimisation, but also the guarantees concerning the existence of a unique invariant measure. This is not a trivial task and many familiar strategies fail. Specifically, our principal contribution in this paper is to develop a framework for reasoning about fairness in social sensing in the sense of guaranteeing that the number of queries per participant will be equalised among comparable participants, in expectation, even when the population of participants varies over time. A prerequisite for fairness is predictability in the sense of guaranteeing that the expected number of queries per participant is independent of the initial state.

Various notions of fairness could then be devised and enforced by shaping the so-called unique invariant measure for a related stochastic system, although we demonstrate only the use of one particular notion of fairness in this paper.

In particular, we develop a meta-algorithm for social sensing in  a time-varying setting, for which we prove guarantees of predictability and fairness by reasoning about the existence of a unique invariant measure for a related stochastic system. We believe that our work is one of the first to deal with this problem in a social-sensing context.\newline

{\bf Comment:} Before proceeding, we remark that the aforementioned notions of ergodicity and invariant measure are very precise technical terms grounded in the theory of stochastic processes. However, we emphasize that our interest in these concepts is entirely practical. Establishing conditions for an invariant measure is the tool we use to prove that a particular algorithm is ergodic. Ergodicity implies two basic things in the context of a stochastic system. First, it means that the initial conditions do not matter in the long run. Thus, it imbues the notions of predictability and fairness in the context of any given algorithm, both of which are important for writing economic contracts. Second, it means that simulations can be trusted. In many cases, the algorithms that we develop give rise to complicated stochastic systems and simulations are often our only tool to understand them. Ergodicity implies that such simulations can be trusted.\newline

\subsection{A motivating application}

\begin{figure}[tb]
\centering
\includegraphics[width=0.98\columnwidth]{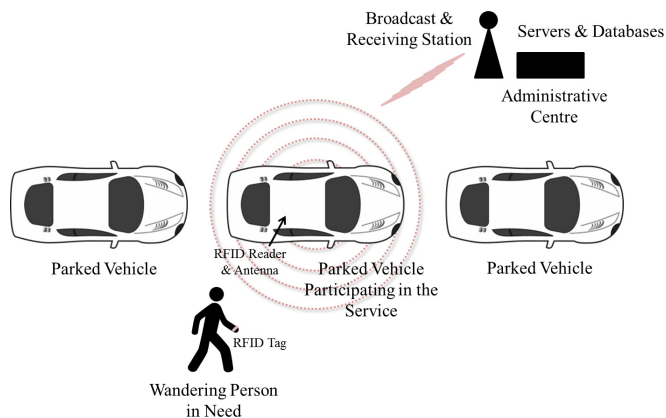}
\caption{An illustration of the RFID-based system, following \cite{souza_conference}. (Some sub-images obtained from Openclipart \cite{OpenClipArt1,OpenClipArt2}.)}\label{illustration}
\end{figure}

As a motivating application, consider the situation when an entity, such as a material object or a pet, goes missing. 
Considering that objects, pets, and even people go missing every day, there exist methods and systems in place to facilitate the location of missing entities.
For example, computer applications allow us to track missing or stolen smartphones. Pets can be micro-chipped or equipped with ``smart'' collars. Medical jewellery and community support networks exist to aid people with needs who wander, including Alzheimer's patients \cite{AlzheimersAssociationIII}, find their ways back home. Looking towards the future, the emergence of the Internet of Things (IoT) may allow for the automation of the search, and therewith, improved response times.

For instance, in the context of IoT, the vehicles that we drive are becoming connected to each other, to the infrastructure, and to the Internet \cite{Luetal2014,Cogilletal2014}. With expanding on-board sensing, computing, and communication capabilities, parked cars no longer need to be idle, to be of no service to us during the extended periods when they are not being driven.  Recently, \cite{Cogilletal2014,Veragoetal2015,Griggsetal2018},
the use of networks of parked vehicles in dense urban areas has been suggested for the detection and localisation of moving, missing entities using RFID technology \cite{landaluce2020review}. From these preliminary studies, a key question arises: How can we distribute the searching agents to quickly locate the moving, missing entity, while also reducing the redundancy, and thus increasing energy efficiency in the system?  
One approach to designing such a system would be as follows. Technically, we consider a feedback loop wherein the administrative centre broadcasts a signal to all agents capable of participating in social sensing. The agents respond to the signal stochastically and thereby alter the state of the system. The administrative center observes a filtered aggregate state of the system, and the process repeats.
In the example of an urban centre with the aim of regulating the number or density of cars looking for a missing,
moving entity efficiently, the administrative center may be the municipality or police force. Probabilistic models of each vehicle switching on or off their RFID reader correspond to the numbers of each vehicle's neighbours that are also capable of participating, and this is obtained by sending out a ``ping'' and observing the responses. The agent uses the broadcast signal from the administration centre, together with the relevant probability model deduced by the number of his or her neighbours, to ``flip a coin'' and determine whether to ``Switch On'' its RFID reader over the next time interval or not. This process is repeated for subsequent time steps. 

\subsection{Specific contributions}

This paper builds on three pieces of our earlier work  \cite{Griggsetal2018,souza_conference,Fioravanti18}. The first paper \cite{Griggsetal2018} concentrated on simulating a specific case of locating an Alzheimer's patient in inner-city Dublin using RFID readers installed in parked cars. That first work considered all parked cars within a certain area would be awoken by a central authority and attempt to locate the missing entity. While simulations have shown that this approach can be highly effective in finding a missing entity, turning all cars on for detection within a given time interval may deplete the participants' batteries without improving coverage substantially, as clustering cars have similar information to report at that given time frame. Moving towards a more efficient way of turning cars' sensors on or off, the second paper \cite{souza_conference} addresses this problem by proposing a stochastic algorithm for regulating the number of cars turned on for searching in which each car would randomly decide its status based on the number of detected neighbouring vehicles. Indeed, cars with many neighbouring vehicles may be less likely to turn their sensors on as it is probable that at least one of the cars in the cluster will turn its sensor on. On the other hand, cars with fewer neighbouring vehicles must have their sensors almost always on as there is a low probability that they can rely on their neighbours to cover that specific area. Simulations have pointed out that similar results to the ones provided in the first paper can be achieved with significantly fewer `active' vehicles at each time frame. The foundation of this stochastic algorithm is described in \cite{Fioravanti18} and is based on a feedback model, in which a population of agents is regulated by a central authority, the controller. In that reference, the ergodicity of the closed-loop system dynamics is sought, as this rather theoretical property is closely related to predictability and fairness in practice. This paper further extends these papers by revisiting the algorithmic aspects of the design, and by extending the application scenario to a more extensive use case. As before, we are interested here in the ergodicity of the feedback system as a prerequisite to writing contracts and to provide fairness guarantees to individual agents. Specifically, the work presented in this current paper expands upon our previous work significantly, as follows.\newline
\begin{itemize}
    \item[A.] We present a framework for reasoning about the predictability and fairness of regulating task distribution in social sensing. We introduce iterated function systems as a tool for addressing such issues in the context of social sensing.\newline
    \item[B.] We develop conditions that ensure predictability and fairness, even when there are small deviations in the probabilistic models over time (this scenario is not considered in \cite{Fioravanti18}).\newline
    \item[C.] We expand upon the motivating application of searching for a missing person with illustrations from simulations from Melbourne, Australia. The new simulations corroborate our analysis: using Algorithm \ref{algo01}, ``Switching On'' or ``Off'' of the RFID readers per participant over time is, indeed,  independent of the initial state and does exhibit weak convergence \cite{van1996weak}.\newline
\end{itemize}
In terms of contributions to the field of social sensing, to the best of our knowledge, this work is the first to develop analytical tools for the design of social sensing platforms that guarantee predictability and certain notions of fairness. This builds on the use of iterated function systems to both model and design such systems. A further contribution to the general field is to use these systems to study the robustness of these systems to uncertainties, and to apply them to the design of a social sensing platform.

\subsection{Paper Organisation}

The paper is structured as follows. In Section II, we provide an overview of related work both in social sensing and control theory.
In Section III, we formalise the problem of predictability and fairness mathematically and present our main 
algorithmic results. In Section IV, we develop a mathematical framework for establishing the robust ergodic properties of social sensing systems.
Finally, in Section \ref{sec:results}, we demonstrate the utility of our results by revisiting the use case of a missing Alzheimer's patient.
Conclusions and future work are presented in Section \ref{conclusions} and supplementary mathematical results are presented in the Appendix.\newline

\section{RELATED WORK}
\subsection{Social sensing}
There exists extensive literature on social sensing, as surveyed in \cite{8666667,wang2015social}. Much of the early work has been empirical and exploratory in nature; e.g., \cite{howe2006rise,artikis2014heterogeneous}. More recently, however, rigorous analyses have started to appear. The authors of \cite{6517107} consider credibility estimation and \cite{restuccia2017quality} set the study in context. \cite{7456345} proposed various likelihood-based inference algorithms 
and bound their performance.
\cite{7932851} combined both efforts in a time-sensitive setting. 
A number of studies \cite{6529080,To2014,7501846,WANG2016,8207350,8351913,8909376,9105079} analysed privacy in this context, as surveyed in \cite{8080202}: \cite{6529080} consider k-anonymity and \cite{7501846,8207350,8909376} consider differential privacy, for instance.
A related stream of work considers distributed and secure storage technologies such as blockchain \cite{8482313} and location-privacy therein \cite{8926541}.
This hints at the maturity of the field.\newline

Numerous applications have been developed, ranging from the
detection of pot holes \cite{eriksson2008pothole},
and crowdsourced traffic monitoring (Nericell \cite{mohan2008nericell}, Waze, or Google Live Traffic \cite{jeske2013floating}), 
road-traffic delay estimation (Waze, VTrack, or Google Live Traffic \cite{jeske2013floating}),
understanding of traffic accidents \cite{artikis2014heterogeneous},
pollution \cite{stevens2010crowdsourcing}, 
and generation of fine-grained noise maps \cite{rana2010ear, maisonneuve2010participatory}, 
to 
the search for missing entities \cite{Cogilletal2014, Veragoetal2015, 6037138,7035659,7929047,Griggsetal2018} such as stolen bicycles \cite{6037138}, lost children \cite{7035659,7929047}, and Alzheimer's patients \cite{Griggsetal2018}. Most recently, social sensing has found applications in sensing within the COVID-19 pandemic \cite{rashid2020covidsens}. Indeed, most track-and-trace approaches, e.g., \cite{Chang2020}, can be seen as a form of social sensing.   
We refer to \cite{chatzimilioudis2012crowdsourcing} for a nice overview of the classical applications.\newline

In this paper, our focus is on the search for missing entities \cite{Cogilletal2014, Veragoetal2015,Griggsetal2018}, 
where, similarly to other vehicle-based approaches \cite{mohan2008nericell,8482313}, task allocation impacts participants' automotive batteries, where the adverse impact is small, but measurable.
Our techniques can be applied  more broadly, though. In general, applications of social sensing, wherein allocated tasks may have an adverse impact on the participants, however small, may benefit from fairness considerations the most. Consider, for instance,  
requiring the driver of a vehicle to focus on the small screen, 
which may impact road safety,
or repeated queries concerning symptoms in medical applications, such as queries as to whether their heads ache in \cite{rashid2020covidsens}, 
where such ``priming'' may change perceptions of the symptoms.

\subsection{RFID-based approaches for the search for missing entities}

Regarding the specific example of agents searching for missing entities, the RFID-based system described in \cite{Cogilletal2014, Veragoetal2015, Griggsetal2018}, and illustrated in Fig. \ref{illustration}, was envisioned as follows. Each participating parked vehicle comprised of an RFID reader and an antenna on board, and was able to communicate with an administrative centre.
The potentially missing entity was presumed to be carrying an RFID passive tag via some means; e.g., a wrist band.
If the entity went missing, an alarm would be raised with the administrative centre. For example, the entity's carer or owner places a phone call to the police. Once the alarm is raised, the administrative centre prompts the application on board the parked vehicles. The RFID technology enables those vehicles to attempt to locate the missing entity, and to inform the administrative centre when it is found (i.e., when the RFID equipment on board a parked vehicle detects and processes the presence of the unique RFID passive tag carried by the missing entity).
Finally, once detected, the administrative centre then invokes a procedure aimed at making contact with the missing entity. For example, the police may go to the location at which the entity was detected to refine the localisation and, if required, help the entity on its way home. See \cite{Cogilletal2014, Veragoetal2015, Griggsetal2018} for further details and \cite{landaluce2020review} for a survey of related work in social sensing with RFID technologies. Investigations conducted in \cite{Griggsetal2018} concerning demonstrating the efficacy of the RFID-based system were purely simulation-based. The system was demonstrated through a use case scenario of a missing Alzheimer's patient in inner-city Dublin, Ireland.
For the simulations, system parameters were varied, including: (i) the percentage of parking spaces on the map of Dublin that were inhabited by vehicles participating in the service; (ii) the polling rate of the RFID equipment on board the participating parked vehicles; and (iii) the RFID equipment's detection range. Results were presented from thousands of simulations and consisted of: (a) the average time that it took for the network of participating parked vehicles to detect the moving pedestrian; and (b) the number of times that the system failed to detect the pedestrian within a thirty-minute time frame. An interesting (albeit expected) observation that the results revealed was one of redundancy, in that the average detection times, and the ``failed to detect'' totals, followed curves resembling the exponential. That is, the average detection time, and ``failed to detect'' results, remained relatively constant until a ``threshold'' participation percentage was reached. When the number of parking spaces inhabited by searching vehicles fell below these thresholds, the detection times, and especially the ``failed to detect'' totals, increased sharply. These observations of redundancy over certain thresholds led to the question of how to distribute the searching agents to quickly locate a moving, missing entity, while also reducing redundancy (and thus increasing energy efficiency) in the system; and thus inspired the preliminary work of \cite{souza_conference}, in which the application of control theory to the motivating problem was first considered.

\subsection{Ergodic properties of the associated closed-loop systems}
The issue of ergodicity has also become topical in the area of  control theory \cite{marevcek2015signalling,WANG2016,marevcek2016r, Fioravanti17, Fioravanti18, kabir2020receding}. Indeed, some of the work presented in this present manuscript builds on \cite{Fioravanti17,Fioravanti18}. These papers develop an abstract framework, blending practical aspects of intelligent transportation systems, smart cities, and techniques from classical control theory. To see the connections to this work, consider a resource allocation problem in discrete time. In particular, consider the closed-loop system as depicted in Fig. \ref{system}, which comprises a (typically large) number of agents, a controller, and a filter. The controller, $\mathcal{C}$, broadcasts a signal $\pi(k)$ at time $k$; the $N$ agents amend  their use of a shared resource in response. The use $x_i(k)$ of the resource by agent $i$ at time $k$ is modelled as a random variable, as there is an inherent randomness in the reaction of each agent to the broadcast signal. The main design task is to regulate the aggregate resource utilisation
$y(k)$, which sums the random variables $x_i(k)$ modelling the individual allocation across all agents $i$ at the given time $k$.
 In this setting, the controller usually does not have access to either
$x_i$ or $y$, but only to an estimate $\hat y$ of $y$, which is the output of a filter $\mathcal{F}$.
\begin{figure}[tb]
\centering
\includegraphics[width=0.92\columnwidth]{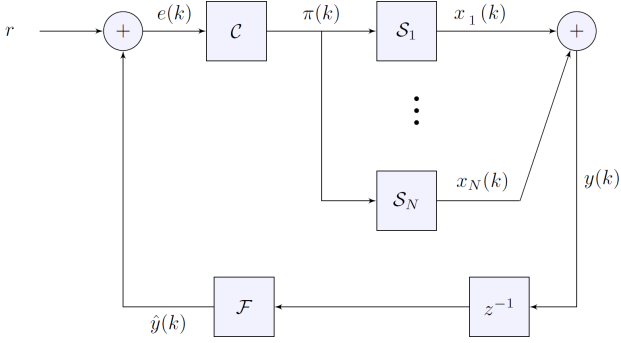}
\caption{A feedback model employed in \cite{souza_conference} and here.}\label{system}
\end{figure}
In addition to achieving regulation, the controller should also ensure that the agents have a sense of
fairness and predictability. In control-theoretic terms, this can be cast as a particular flavour of the ergodicity of the closed-loop system dynamics, known as the existence of a unique invariant measure \cite{Fioravanti17,Fioravanti18}.
This completely removes the effects of initial conditions in the long run. Overall, in the aforementioned references, the authors 
state the conditions for the unique ergodicity of the closed loop with linear controllers and filters:
\begin{thm}\cite[Theorem 3]{Fioravanti17} \label{thm:theorem-one}
Consider the feedback system depicted in Fig. \ref{system}, for some given finite-dimensional linear systems ${\mathcal C}$ and ${\mathcal F}$.
Assume that each agent $i \in \{1,\dots,N\}$ has state $x_i(k)$ governed by the following affine stochastic difference equation:
\begin{equation}\label{eq:irf_sys}
x_i(k+1) = w_{ij}\left(x_i(k)\right),
\end{equation}
where the affine mapping $w_{ij}$ is chosen at each step of time according to a Dini-continuous
probability function $p_{ij}(x_i(k), \pi(k))$, out of
\begin{align}\label{eq:func-1}
w_{ij}(x_i)= A_i x_i + b_{ij}    
\end{align}
where $A_i$ is a Schur matrix and for all $i$, $\pi(k)$, $\sum_j p_{ij}(x_i(k),\pi(k)) = 1$. In addition, suppose that there exist scalars $\delta_i > 0$ such that $p_{ij}(x_i,\pi)
\geq \delta_i > 0$; that is, the probabilities are bounded away from
zero. Then, for every stable linear controller $\mathcal{C}$ and every
stable linear filter $\mathcal{F}$, the feedback loop converges in
distribution to a unique invariant measure.
\end{thm}
This theoretical framework will be exploited and extended in the sequel to devise our social sensing solution. For now, there are some key aspects on this framework and specifically on the previous theorem that should be pointed out and discussed. Note first that the agents' dynamic behaviour may seem rather limited, but it suffices for several smart-cities applications, such as applications with `on-off' participants; the reader may see \cite{Fioravanti18} for extensions to the nonlinear case. Note also that the main design task in the linear setting described above is to devise two stable linear time-invariant systems (a filter and a controller) so that the closed-loop dynamics are stable. This ensures ergodicity and, thus allows for fairness. Finally, it is important to point out that the probabilities involved in the dynamic response of the agents with respect to the broadcast signal must be bounded away from zero. The lack of this assumption can yield non-ergodic stochastic processes, and in this case some agents may monopolise allocated resources.

\section{PROBLEM STATEMENT}\label{section-problem-statement}
Let us revisit the closed-loop schema of Fig. \ref{system}, where there are $N  \in \N$ agents $\mathcal{S}_1, \dots, \mathcal{S}_N$, whose aim is to estimate the state evolution of some underlying system. 
These $N$ agents are regulated by a controller $\mathcal C$ using broadcast signal $\pi(k)$, at time $k \in \N$, which affects the agents' participation in the sensing scheme.
The state of each agent $i$ at time $k$ is captured by $x_i(k)$, which can be univariate or multivariate.

For example, the state $x_i(k)$ at time $k$ could be in the set $\{0,1\}$, which would suggest whether agent $i$ allows for the participation in social sensing ($x_i(k) = 1$) or not ($x_i(k) = 0$). In this case, at each time instant $k$, agent $i$ may have a probability $p_{i1}$ of being on and a probability $p_{i0}$ of being off at the following time step.  Both probabilities may depend on the broadcast control signal $\pi$; that is,
\begin{equation}\label{eq:state-dep-prb1}
\mathbb{P}(x_i(k+1) = 1) = p_{i1}\left(\pi(k)\right)
\end{equation}
and, thus,
\begin{equation}\label{eq:state-dep-prb2}
\mathbb{P}(x_i(k+1) = 0) = p_{i0}\left(\pi(k)\right) = 1 - p_{i1}\left(\pi(k)\right),
\end{equation}
since both events are complementary. 
More generally, we could consider a family of response functions $\{ w_{\sigma} \}_{\sigma=1}^{N}$, with probability functions $\{p_{\sigma}(x)\}_{\sigma=1}^{N}$ where $$p_{\sigma}(x): \mathcal X\to [0,1] ,\sum\limits_{\sigma=1}^{N} p_{\sigma}(x)=1,$$ and agent $i$ selects $\sigma$ according to the state-dependent probabilities
\begin{align}\label{eq:time-invariant-p}
\mathbf{p}^i(x^i(k))=(p_1^i(x^i(k)),\dots, p_N^i(x^i(k))).
\end{align}
While there is some inherent randomness in the reaction of each agent to the broadcast signal,
increasing the value of $\pi$ should increase the probabilities of participation and likewise lowering $\pi$ should induce agents to stop participating, eventually.

Finally, based on the participation of the agents, the controller has access to a filtered aggregate of observations $\hat y(k)$, with some delay ($z^{-1}$),
possibly after subtracting a reference value $r$ to obtain the error $e(k)$.
The error $e(k)$ is then used to produce the broadcast signal, thus closing the loop.
See, also, Algorithm \ref{meta} on Page \pageref{meta}.

Informally, predictability  requires that, for each agent, there exists a limit on the long-run average of the agent's state, and that this limit is independent of the agent's initial state. 
Fairness, consequently, requires that this limit coincides for all agents.
Formally:
\begin{dfn}[Predictability]\label{dfn:pred}
Whenever, for each agent $i$, there is an agent-specific constant $\overline{r}_i$ such that the following limit exists:
\begin{equation}\label{eq:predi}
\lim_{k\to \infty} \frac{1}{k+1} \sum_{j=0}^k x_i(j) = \overline{r}_i, \quad \textrm{a.s.},
\end{equation}
i.e., a long-run average of agents' states independent of the initial state $x_i(0)$, we say the system is \emph{predictable}. 
\end{dfn}
Next, fairness in the sense of statistical parity \cite{dwork2012fairness}
requires the limits of \eqref{eq:time-varying-fair} coincide for all agents $i$.
\begin{dfn}[Fairness]\label{dfn:fair}
Whenever there exists a finite constant $\overline{r}$ such that:
\begin{equation}\label{eq:time-inv-fair}
\lim_{k\to \infty} \frac{1}{k+1} \sum_{j=0}^k x_i(j) = \overline{r}, \quad \textrm{a.s.},
\end{equation}
for all agents $i$, we say that the system is \emph{fair}.
\end{dfn}
Notice that this notion of fairness is rather strict. One may equally well consider simpler notions of fairness,
perhaps summing over only certain coordinates of the multivariate 
state variable, or considering a fixed numerical threshold:
\begin{dfn}[$\epsilon$-fairness] \label{dfn:eps-fair}
Based on \eqref{eq:predi} and \eqref{eq:time-inv-fair}, we define the predictability and fairness vectors for some $r\in \mathbb R$ as follows:
\begin{align}\label{eq:pred-vec}
\hat p&=(\overline r_1, \overline r_2,\dots, \overline r_n)^{\top}\in \mathcal X\subseteq \mathbb R^n,\
\end{align}
\begin{align}\label{eq:fair-vec}
\hat f&=r.\mathbf{1}^{\top}, \textrm{ where }  \mathbf{1}^{\top}=\left(1,1,\dots, 1\right)^{\top}    
\end{align}
and, for some small $\epsilon>0$ and any vector norm $\|\cdot\|$ in $\mathbb R^n$, we say that the system is \emph{$\epsilon$-fair} if 
we have $\mathbb E\left(\|\hat p - \hat f\|\right)\le \epsilon$.
\end{dfn}
Note that this definition does not imply the existence of a protocol that could ensure $\epsilon$-fairness of the system for any input. Indeed, however large the value of $\epsilon$, and however fast the convergence of the algorithm schema for social sensing, 
a scenario can be created to violate the $\epsilon$-fairness.
Furthermore, note that while in some smart-cities applications, one may assume that the  probabilistic model is time-invariant, most social-sensing problems feature a time-varying population, and this can be challenging from a theoretical perspective. For instance, in our application,  the missing entity may move quickly (e.g., using the underground) and the number of parked cars in each agent's surroundings may change slowly.

In such a time-varying setting, two complications arise. First, the efficient task allocation \cite{boutsis2014task,duan2017distributed} (e.g., search efficiency in our motivating application) becomes computationally intractable \cite{7365407} when the probabilistic models are allowed to vary arbitrarily. 
More formally, the approximation to any non-trivial factor with respect to the number of queries is complete for polynomial-space Turing machines \cite{Papadimitriou1999}.
(This relies on the equivalence with the so-called restless multi-armed bandit problem \cite{whittle1988restless,weber1990index}, a well-known problem in applied probability.)
Hence, any social sensing scheme assuring search efficiency is computationally complex, independent of whether P equals NP, 
and in turn, 
predictability and fairness are as much as we can hope for.
Second, the analysis of predictability and fairness becomes rather non-trivial.
We address these complications with tools from stochastic analysis and control. 
Both predictability and fairness  
can be defined in terms of the properties of an associated stochastic model, which is known as an iterated function system (cf. Definition \ref{dfn:ifs} in the next section).
When the probabilistic model does not change over time, 
predictability is assured by the existence of a unique invariant measure (cf. Definition \ref{dfn:invariant} in the next section).
When we cannot rely on the probabilities and transformations in the iterated function system being invariant over time, or perfectly known to us, there are still at least two options. 
Either we can consider the notion of piece-wise stationary measures \cite{ghosh2019iterated} 
for a time-varying iterated function system \cite{ghosh2019iterated}, or we can consider perturbation analysis, also known as sensitivity analysis.
There, it is of interest to know whether a perturbation in the state causes a large difference in the behavior of the corresponding stochastic process. In our class of contractive transformations, 
we show in Theorem \ref{thm:purturbation} that small perturbations in the state or probabilities do not cause large changes in the behavior, in terms of the long-run average state. This means that we can use linear or other approximations without changing the invariant measures too much.

\section{MAIN RESULTS}\label{ref:theory}
\begin{algorithm}[t!]
\caption{An algorithm schema for social sensing with fairness guarantees}
\KwData{Number of agents $N$;  initial state $x^i(0)\in \mathcal X$ for each agent $i$; a set of possible behaviours $\{w_\tau\}_{\tau}$ valid for any agent, to be chosen with agent- and state-dependent probability; number $t$ of time steps between perturbations; time horizon $t \le T$ of time steps; a bound $\delta$ on the rate of the environment-driven change per $t$ time steps.}
{\bf Initialise} counters $s \leftarrow 0, h \leftarrow 0$, where $(s, h)$ considered lexicographically captures time \;
Central authority {\bf broadcasts} arbitrary signal $\pi(0)$, such as $0$ \;
\While{$s \cdot h \leq T$}{
\While{$h \leq t$}{
\For{each agent $i$}{
Agent $i$ {\bf calculates} state-dependent probabilities $\mathbf{p}^i(x^i(st + h))=(p_1^i(x^i(st + h)),\dots, p_N^i(x^i(st + h)))$.\;
Agent $i$ {\bf selects} response function $w_{\sigma}$, where $\sigma$ is chosen according to to the probabilities
$\mathbf{p}^i(x^i(st + h))$\;
Agent $i$ {\bf updates} state $x^i(st + h+1)$ using $x^i(st + h+1)=w_{\sigma_h^i}(x^i(st + h))$, i.e., according to \eqref{eq:state-dep-prb1}\;
}
Central authority {\bf observes} filtered aggregate state $\hat y(st + h)$, where the filter ${\cal F}$ is possibly not known a priori \;
Central authority {\bf computes} the error $e(st + h)$ \;
Central authority {\bf broadcasts} signal $\pi(st + h)$ computed using some controller ${\cal C}$ and increments $h$ to $h +1$\;
}
The environment {\bf perturbs} the state of agents such that $| x^i((s+1)t) - x^i(st + h) | \le \delta $ and increments $s$ to $s +1$. 
}
\label{meta}
\end{algorithm}

To address predictability and fairness in social sensing rigorously, we present a result which is applicable for a class of stochastic phenomena which can be modelled as iterated function systems. This result then makes it possible to model small variations in the response of the participants within a social sensing scheme, which is a stepping stone towards analyses of a time-varying response of a population.

\subsection{A class of stochastic systems}
Let us define the class of systems we consider formally:
\begin{dfn}[Iterated function system with state-dependent probabilities \cite{barnsley1993}]
\label{dfn:ifs}
Let $\mathcal X\subseteq \mathbb R^n$ be closed, and let $\rho$ be a metric on $\mathcal X$ such that $\left(\mathcal X, \rho\right)$ is a complete metric space. Let $\{w_{\sigma}\}_{\sigma=1}^{N}$ be transformations on $\mathcal X$ and $\{p_{\sigma}(x)\}_{\sigma=1}^{N}$ be probability functions defined on Borel sigma-algebra $\mathcal B(\mathcal X)$, such that, for all $\sigma\in [1,N]$, 
\begin{align*}
&p_{\sigma}(x):\mathcal X \to [0,1] \quad \text{and } \sum\limits_{\sigma=1}^{N} p_{\sigma}(x)=1.
\end{align*}
The pair of sequences
\begin{align}\label{ifsdef}
\left( w_1(x),w_2(x), \dots, w_N(x); p_1(x), p_2(x),\dots, p_N(x)\right)
\end{align} 
is called an iterated function system (IFS).
\end{dfn}
Informally, the corresponding discrete-time Markov process on $\mathcal X$ evolves as follows:
Choose an initial point $x_0\in \mathcal X$. Select an integer from the set $[1, N]:=\{1,2,\dots, N\}$ in such a way that the probability
of choosing $\sigma$ is $p_{\sigma}(x_0)$, $\sigma\in [1,N]$. When the number $\sigma_0$ is drawn, define
\begin{align*}
x_1= w_{\sigma_{0}}(x_0).    
\end{align*}
Having $x_1$, we select $\sigma_1$ according to the distribution 
\begin{align*}
p_1(x_1), p_2(x_1), \dots, p_N(x_1),    
\end{align*}
and we define 
\begin{align*}
x_2= w_{\sigma_1}(x_1),
\end{align*}
and so on.

Let us denote $\nu_k$ for $k=0,1,2,\dots$, the distribution of $x_k$, i.e.,
\begin{align}\label{distri_xn}
\nu_k(\mathcal A)= \mathbb P(x_k \in \mathcal A) \text{ for some } \mathcal A \in \mathcal B(\mathcal X).
\end{align}
The above procedure can be formalized for a given $x\in \mathcal X$ and a Borel subset $\mathcal A\in \mathcal B (\mathcal X)$, we may easily show that the
transition operator for the given IFS is of the form:
\begin{align}\label{transii}
\nu(x, \mathcal A):= \sum_{\sigma=0}^{N} 1_{\mathcal A}\left(w_{\sigma}(x)\right) p_{\sigma}(x),
\end{align}
where $\nu(x,A)$ is the transition probability from $x$ to $\mathcal A$ and where $1_{\mathcal A}$ denotes the characteristic function of $\mathcal A$:
\begin{equation}\label{eq:char-func}
1_{\mathcal A}(x):=
\begin{cases} 
      1 & \text{ if } x\in \mathcal A.\\
      0 & \text{ if } x\in  \mathcal A^c.
   \end{cases}
\end{equation}
where, in turn, $\mathcal A^c$ denotes the complement of the event or the Borel subset $\mathcal A$.
\begin{dfn}[Markov operator \cite{myjak2003}]\label{def:markov}
Closely connected with this transition probability is the \emph{Markov operator}, denoted by $P$, defined on the space of all real or complex valued Borel measurable maps $f$ on $\mathcal X$ as:
\begin{align}\label{markov}
P f(x)=\int_{\mathcal X} f(y) \nu(x, d y)
\end{align}
\end{dfn}
\begin{dfn}[Invariant probability measure \cite{BarnsleyDemkoEltonEtAl1988,myjak2003}]
\label{dfn:invariant}
If a Markov chain $\{X_k\}$ moves with transitional probability $\eqref{transii}$, then it is of great interest to know the existence of an \emph{invariant probability measure} for the chain, i.e., existence of a probability measure $\nu_{\star}\in \mathcal M (\mathcal X)$, for which:
\begin{align}
\nu_{\star}(\mathcal A)= \int_{\mathcal A} \nu(x, \mathcal A) \nu_{\star}(dx) \quad \forall \mathcal A \in \mathcal B (\mathcal X).
\end{align}
In our analytic approach, we consider the dual of \emph{Markov operator} $P$ defined in \eqref{markov}, 
\begin{align}\label{analytic-dual}
(P^{\star}\nu)(\mathcal A)=\int_{\mathcal X} \nu(x, \mathcal A) \nu(dx),    
\end{align}
a map defined on the space of all Borel measures on $\mathcal X$.
A probability measure $\nu_{\star}$ is called  \emph{invariant probability measure} for the Markov chain $\{X_k\}$ with Markov operator $P$ if and only if
\begin{align}
P^{\star}\nu_{\star}(\mathcal A)=\nu_{\star}(\mathcal A) \quad \forall \mathcal A \in \mathcal B (\mathcal X).
\end{align}
\end{dfn}

We finish this preliminary section of mathematical definitions and set up by defining a useful metric on the space of probability measure on $\mathcal X$, due to Kantorovich and Rubinstein \cite{Evans99partialdifferential, mass1998, Villani2009OptimalT}, also known as Wasserstein-$1$ distance.
\begin{dfn}[Wasserstein-$1$ distance; Remark 6.5, p. 95 in \cite{Villani2009OptimalT}]
Let $\mathcal{L}_1$ denote the space of all Lipschitz maps with Lipschitz constant $1$, i.e
\begin{align*}
\mathcal L_1=\{f: \mathcal X\to \mathbb R : |f(x)-f(y)|\le \rho (x,y) \quad \forall x, y\in \mathcal X\}.
\end{align*} For $\nu_1, \nu_2\in \mathcal M (\mathcal X)$, Wasserstein-$1$ distance between these two probability measure is denoted by $\mathcal W_1(\nu_1, \nu_2)$ and is given by:
\begin{align}\label{eq:wass-dist}
\mathcal W_1(\nu_1, \nu_2)=\sup\limits_{f\in \mathcal L_{1}}\left[\int fd\nu_1-\int f d\nu_2\right].
\end{align}
\end{dfn}

\subsection{The main result}
For a given iterated function system, the following result provides a bound on the distance between its original invariant measure and the invariant measure obtained when it is perturbed. From a practical viewpoint, this bound ensures that small perturbations in the IFS parameters do not cause large changes to its long-run behaviour.
\begin{thm}
\label{thm:purturbation}
Let $P_1^{\star}$ be the Markov operator \cite{myjak2003}  of an iterated function system $\left(w_{1}(x),\dots, w_N(x); p_{1}(x), \dots, p_N(x)\right)$ with invariant measure $\nu_1^{\star}$, and let $P_2^{\star}$ be the Markov operator of the perturbed iterated function system $\left(w'_{1}(x),\dots, w'_{N}(x); p'_{1}(x), \dots, p'_{N}(x)\right)$ with invariant measure $\nu_2^{\star}$. Then, for all $x \in \mathcal{X}$, we have the following estimates of distance between the invariant measures in Wasserstein-$1$ distance (\ref{eq:wass-dist}):
\begin{align}\label{eq:esti-1}
&\mathcal W_1(\nu_{1}^{\star},\nu_{2}^{\star})\nonumber\\
&\le \frac{1}{1-r}\left(r'\sum\limits_{\sigma_{k}} p_{\sigma_{k}}(x)\Big\|w_{\sigma_{k}}(x) - w'_{\sigma_{k}}(x)\Big\|_{\infty}+\beta \eta\right)
\end{align}
where 
\begin{itemize}
\item[(a)] $0<r<1$, is a contraction factor for the Markov operator in $\mathcal W_1$ metric, 
\item[(b)] $\sigma_{0},\sigma_{1},\sigma_{2},\dots$ are i.i.d discrete-random-variable taking values in $\{1,2,\dots, N\}$,
\item[(c)] $\eta$ is the bound on the perturbation in probabilities, i.e.,
$\sum\limits_{\sigma_k} \lVert p_{\sigma_k}(x)-p'_{\sigma_k}(x)\rVert \le \eta$,
\item[(d)] $\beta$ is a bound for the real-valued continuous function $w\in C_b(\mathcal X, \mathbb R)$,
\item[(e)] $\|w(x)-w(y)\|\le r'\|x-y\|$, for some $r'$ and for almost all $x, y\in \mathcal X$.
\end{itemize}
\end{thm}
The proof is included in the Appendix.
\subsection{A perturbation analysis in the time-invariant setting}
Theorem \ref{thm:purturbation} makes it possible to reason about small changes to the behaviour of participants in a social sensing scheme. 
This can be seen as a perturbation analysis or stability analysis for iterated function systems:
\begin{cor} \label{cor:one}
Consider the feedback system depicted in Fig. \ref{system}, for some given finite-dimensional linear systems ${\mathcal C}$ and ${\mathcal F}$.
Assume that each agent $i \in \{1,\dots,N\}$ has state $x_i(k)$ governed by the following affine stochastic difference equation \eqref{eq:irf_sys} where the affine mapping $w_{ij}$ is chosen at each step of time according to a Dini-continuous
probability function $p_{ij}(x_i(k), \pi(k))$, out of \eqref{eq:func-1}.
If the system is perturbed in such a way that the perturbed system is described by
\begin{align}\label{eq:purturbed}
x_i(k+1) = w'_{ij}\left(x_i(k)\right),    
\end{align}
where the affine mapping $w'_{ij}$ is chosen at each step of time according to a Dini-continuous
probability function $p'_{ij}(x_i(k), \pi(k))$, out of
\begin{align}\label{eq:func-2}
w'_{ij}(x_i)= A'_i x_i + b'_{ij}.    
\end{align}
Then, if $P_1^{\star}$ be the Markov operator for the system \eqref{eq:irf_sys} with invariant measure $\nu_1^{\star}$, and if $P_2^{\star}$ be the Markov operator of the perturbed system \eqref{eq:purturbed} with invariant measure $\nu_2^{\star}$, we have, for all $x \in \mathcal{X}$, the following estimates of distance between their invariant measure in Wasserstein-$1$ distance:
\begin{align}
&\mathcal W_1(\nu_{1}^{\star},\nu_{2}^{\star})\nonumber\\
&\le \frac{1}{1-r}\left(r'\sum\limits_{\sigma_{k}} p_{\sigma_{k}}(x)\Big\|w_{\sigma_{k}}(x) - w'_{\sigma_{k}}(x)\Big\|_{\infty}+\beta \eta\right)
\end{align}
where 
\begin{itemize}
\item[(a)] $0<r<1$, is a contraction factor for the Markov operator in $\mathcal W_1$ metric, 
\item[(b)] $\sigma_{0},\sigma_{1},\sigma_{2},\dots$ are i.i.d discrete-random-variable taking values in $\{1,2,\dots, N\}$,
\item[(c)] $\eta$ is the bound on the perturbation in probabilities, i.e.,
$\sum\limits_{\sigma_k} \lVert p_{\sigma_k}(x)-p'_{\sigma_k}(x)\rVert \le \eta$,
\item[(d)] $\beta$ is a bound for the real-valued continuous function $w\in C_b(\mathcal X, \mathbb R)$,
\item[(e)] $\|w(x)-w(y)\|\le r'\|x-y\|$, for some $r'$ and for almost all $x, y\in \mathcal X$.
\end{itemize}

\end{cor}
Such a perturbation analysis is also a small step from the time-invariant setting of Theorem~\ref{thm:theorem-one} towards a time-varying setting.
\subsection{A time-varying setting}\label{sub-sec:time-varying}
Many practical applications do, indeed, involve time-varying populations, i.e., populations that change over time. In particular, time-varying response functions of populations are very clearly observable in most real-world social-sensing applications, where people follow diurnal rhythms.
At night, there may be fewer participants, whose responses may be  different from the day-time participants'.  

Let us define the time-varying setting formally.
Let $\mathcal X$ be a closed subset of $\mathbb R^n$. We are given a finite set of bounded Lipschitz transformations:
\begin{align*}
\mathcal L=\{w_{\sigma}: \mathcal X\to\mathcal X\}_{\sigma=1}^{N}
\end{align*} and a countable family of $N$-tuple probability functions
\begin{align}\label{probfunc}
\{\mathbf{p}^s(x)=(p_1^s(x), p_2^s(x),\dots, p_N^s(x))\}_{s=1}^{\infty}    
\end{align}
where the variable $s$ denotes a discrete time-scale, for each fixed $s\in \mathbb N$, and for all $\sigma\in [1, N]$, $p_{\sigma}^s: \mathcal X\to [0,1]$ and for any fixed $s\in \mathbb N$,
\begin{align}\label{condprob}
& 0\le p_{\sigma}^s(x)\le 1\hspace{0.3cm}\forall \sigma\in [1,N], \nonumber\\   
&\sum\limits_{\sigma=1}^{N} p_{\sigma}^s(x)=1\quad \forall x.
\end{align}
We now introduce a time-varying stochastic situation as follows: let $s$ denote a discrete time-scale, between $s=1$ and $s=2$, a certain number say $k=1,2,\dots, t$ iteration is performed for the system \eqref{eq:irf_sys} with a tuple of probability function $\{(p_1^1(x), p_2^1(x),\dots, p_N^1(x))\}$ and after such number of iteration we change the probability-tuple and we perform the iteration again, for a general $s$, time-varying situation \eqref{eq:irf_sys} becomes 
\begin{equation}\label{eq:tvirf}
x^s(k+1)=w_{\sigma_k^s}(x^s(k))
\end{equation}
with the probability tuple $\{(p_1^s(x), p_2^s(x),\dots, p_N^s(x))\}$ ,where $\sigma^s_0,\sigma^s_1,\dots$ are i.i.d discrete-random-variable taking values in $\{1,2,\dots, N\}$.
To aid exposition, let us illustrate this with a simple example. At time scale $s=1$, choose $x^1(0)\in \mathcal X$ and calculate $\mathbf{p}^1(x^1(0))$ as defined in \eqref{probfunc}, we use this vector as the chance of selecting a transformation
\begin{align*}
w_{\sigma_1^1}(x^1(0)), \quad \sigma_1^1\in\{1,2,\dots, N\}     
\end{align*}
This is done by considering the probabilities as bins, $\sigma_1^1$, with length $p_i^1(x^1(0))$.
Placing these bins end to end on $[0,1]$ will fill the interval as a consequence of \eqref{condprob}.
We then choose a random number $q\in [0,1]$ and the bin containing $q$ corresponds to the
probability function we choose. The starting point of the next iteration is
\begin{align*}
x^1(1)=w_{\sigma_1^1}(x^1(0))    
\end{align*}
is calculated and consequently a new probability vector $\mathbf{p}^1(x^1(1))$ must be calculated. 
See Algorithm 1 for the general schema.
For time steps $t$ between $s$ and $s+1$, the conditional distribution of the future depends only on the current state. For the time steps between $s$ and $s+1$, the process defined 
by \eqref{eq:tvirf} is clearly Markovian.\newline
Let $C_b(\mathcal X)$ denote the set of real-valued bounded continuous functions on $\mathcal X$, for any $s$, one can define a linear map $P_s$ on $C_b(\mathcal X,\mathbb R)$:
\begin{align}\label{linear}
P_sw(x):=\sum_{i=1}^{N} p^s_{i}(x)(w\circ w_{\sigma_i^s})(x) 
\end{align}
This operator characterizes the Markov chain. Due to the fact that $P_s$ maps $C_b(\mathcal X)$ into itself, which is also known as $P_s$ being a Feller map \cite[cf. Section 5.1]{santos2005accuracy},
the chain is sometimes called Feller chain.  We will mainly be interested in the problem of uniqueness or non-uniqueness of invariant probability measures. A probability measure $\nu$ on $\mathcal X$ is called invariant for the operator $P_s$ if 
\begin{align}\label{invprob}
\int_{\mathcal X} (P_s w) d \nu=\int_{\mathcal X} w d \nu \quad \forall w\in C_b(\mathcal X).
\end{align}
Denoting the dual of the map $P_s$ as follows:
\begin{align}\label{eq:dual-2}
P_s^{\star}: \mathcal M(\mathcal X)\to \mathcal M(\mathcal X),
\end{align}
with the requirement 
\begin{align}
\int_{\mathcal X} w d (P_s^{\star} \nu)=\int_{\mathcal X} (P_s w) d \nu,
\end{align}
then \eqref{invprob} is reduced as: a $\nu_{\star}\in \mathcal M(\mathcal X)$ is invariant if and only if
\begin{align}
P_s^{\star}\nu_{\star}=\nu_{\star}.
\end{align}
Such a dual map $P_s^{\star}$ is well-defined by the Riesz representation theorem.

\begin{thm}\label{thm:existence}
Let for each $s, P_s$ be defined as in \eqref{linear}. If there exists $x\in \mathcal X$, for which the sequence of transitional probability measures $\{\nu_{m}^s(x, \cdot)\}_{m\ge 0}$ is uniformly tight, then there exists an invariant probability measure for $P_s^{\star}$.
\end{thm}
The proof is included in the Appendix. A number of further results can be shown in this setting, as described in Appendix \ref{sec:C}.

\subsection{Implications}
When the existence of a unique ergodic measure is guaranteed, either by
Theorem \ref{thm:theorem-one}, or by Theorem~\ref{thm:existence} in Appendix \ref{sec:C}, this assures predictability, as introduced in Definition~\ref{dfn:pred}:
\begin{cor}[Predictability in the Time-Invariant Setting] \label{cor:pred-inv}
Consider the feedback system depicted in Fig. \ref{system}, for some given finite-dimensional stable linear systems ${\mathcal C}$ and ${\mathcal F}$. Assume that each agent $i \in \{1,\dots,N\}$ has state $x_i^s(k)$ governed by  \eqref{eq:irf_sys}, where the affine mapping $w_{ij}^s$ is chosen at each step of time according to a Dini-continuous
probability function $p_{ij}^s(x_i^s(k),\pi(k))$, out of
\eqref{eq:func-2}, where $A_i$ is a Schur matrix and for all $i$, $\pi(k)$, and for each fixed $s$, $\sum_j p^s_{ij}(x^s_i(k),\pi(k)) = 1$. Moreover, assume that the probabilities $p_{ij}^s$ are bounded away from zero and that the conditions of Theorem \ref{thm:existence} hold for the process thus defined. Then, the 
feedback loop ensures {\em predictability} of each agent's dynamics, i.e., for each agent $i$, there exists a constant $\overline{r}_i$ such that 
\begin{equation}\label{eq:pred}
\lim_{k\to \infty} \frac{1}{k+1} \sum_{j=0}^k x^s_i(j) = \overline{r}_i, \quad \textrm{a.s}.
\end{equation}
\end{cor}

Corollary \ref{cor:pred-inv} ensures that, under the mild assumptions of Theorem \ref{thm:theorem-one}, the participants' trajectories still couple for different initial conditions; that is, the predictability still holds. As stated before, such a property is important in practical social sensing problems\footnote{This property is also desirable in resource-sharing problems \cite{Fioravanti18}.}, as the central authority thus ensures a predictable task allocation.

In turn, predictability allows for fairness, as introduced in Definition~\ref{dfn:fair}, albeit under strict conditions suggesting that the agent's behaviour is symmetric in some sense and that their initial states are the same:

\begin{cor}[Fairness in the Time-Invariant Setting] \label{cor:fair}
Consider the feedback system depicted in Fig. \ref{system}, and the same conditions as in Corollary \ref{cor:pred-inv}. 
If, in addition:
\begin{itemize}
\item  if the agents' states evolve from a uniform initial state, that is, if there exists a constant $c$ such that $x_i(0) = c$ for all agents $i$,
\item the state-dependent probabilities 
\eqref{eq:time-invariant-p} are uniform, i.e.,
there exists family of functions $
\{q_{\sigma}(x)\}_{\sigma=1}^{N}$,
$q_\sigma: \mathcal X\to [0,1]$ such that for all agents $i=1,2 \ldots N$ and all $x \in \mathcal X$,
$p^i_\sigma(x) = q_\sigma(x)$,
\end{itemize}
then the feedback loop ensures {\em fairness} of the agents' dynamics, that is, there exists a constant $\overline{r}$ such that 
\begin{equation}\label{eq:fair}
\lim_{k\to \infty} \frac{1}{k+1} \sum_{j=0}^k x^s_i(j) = \overline{r}\quad\textrm{a.s.}
\end{equation}
for all agents $i$.
\end{cor}

\begin{proof}[Proof of Corollary \ref{cor:fair}]
Fairness follows immediately from the Markov property of the iterated function system with state-dependent probabilities \eqref{transii}. 
\end{proof}

\begin{rem}
Throughout Corollaries
\ref{cor:pred-inv}--\ref{cor:fair}, Dini's condition on the probabilities may be replaced by simpler, more conservative assumptions, such as Lipschitz or H\"older continuity \cite{BarnsleyDemkoEltonEtAl1988}.
\end{rem}


In Appendix~\ref{sec:D}, we generalise these corollaries to the time-varying setting. 
The question whether further generalisations of fairness, such as $\epsilon$-fairness of Definition \ref{dfn:eps-fair}, allow for less strict conditions on the initial state, are most intriguing, but left open.

\section{THE SEARCH FOR MISSING ENTITIES}
\label{section-algorithm}
We are now able to showcase the application of our abstract algorithmic schema (cf. Algorithm 1) in the context of searching for missing entities utilising a network of parked cars. Algorithm \ref{algo01} specialises Algorithm 1 as follows: abstract agents are specialised to cars; the state is composed of the internal states of the cars ($x^i$) and the numbers of cars parked in their vicinity ($N^i$), as well as the possible states of the controller ($\mathcal C$) and filter ($\mathcal F$). The state $x_i(k)$ at time $k$ is in the set $\{0,1\}$, which models whether agent $i$ participates in the search ($x_i(k) = 1$) or not ($x_i(k) = 0$).

The abstract sets of state-to-state functions $w$ of Algorithm 1 are replaced with three possible behaviours $f_f, f_s, f_m$ of the cars, corresponding to a few, some, and many cars parked in their vicinity, as depicted in Fig. \ref{probabilitycurves_melbourne}.
The abstract agent- and state-dependent probability $\mathbf{p}^i(x^i(t))=(p_1^i(x^i(t)),\dots, p_N^i(x^i(t)))$ of choosing a particular abstract $w$ at time $t$ is specialised to an agent- and state-dependent probability $( p_f^i(t), p_s^i(t), p_m^i(t) )$ of the three possible behaviours. 
These probabilities may depend on the number of neighbouring vehicles, but must satisfy the conditions of Theorem \ref{thm:theorem-one}.
Indeed, clusters of cars can cooperate and take turns to cover one area, whereas a sole car on a street must be almost always on. 

Finally, $k$ is a counter of the samples drawn using any of the three-vectors, since the most recent signal is broadcast. Another counter $h$ counts the number of signals broadcast, since the most recent perturbation of the states. Finally, $s$ is the counter of the perturbations. Time is hence captured by a triple $(s, h, k)$, considered lexicographically.

\begin{algorithm}
\caption{A specialisation of Algorithm 1 for the search for a missing entity.}
\KwData{Number of agents $N$;  initial state $x^i(0)\in \mathcal X$ for each agent $i$; a set of possible behaviours $\{f_f, f_s, f_m\}$ valid for any agent, based on the few, some, or many cars in the vicinity; number $t$ of time steps between perturbations; time horizon $t \le T$ of time steps; a bound $\delta$ on the rate of change of the number $N_i$ of cars parked in the vicinity of car $i$, within $t$ time steps.}
\KwResult{Missing entity location {\bf or} fail alert.}
{\bf Initialise} $s \leftarrow 0$; $h \leftarrow 0$; $\pi(0) \leftarrow 0$; $x_i(0) \leftarrow 0$; $\hat y(0) \leftarrow 0$\;
\While{$s \cdot h \leq T$}{
\While{$h \leq t$}{
\For{each car $i$}{
Car $i$ {\bf determines} the number $N_i(st +h)$ of neighbouring cars\;
Car $i$ {\bf decides} whether $N_i(st +h)$ corresponds to {\em few}, {\em some} or {\em many} neighbouring cars\;
Car $i$ {\bf sets} $\mathbf{p}^{i} = ( p_f^i(st +h), p_s^i(st +h), p_m^i(st +h) )$ corresponding to the behaviours $\{f_f, f_s, f_m\}$ for {\em few}, {\em some} or {\em many} neighbouring cars \;
Car $i$ {\bf ``tosses a coin''} and updates state $x^i(st +h)$ using one of $\{f_f, f_s, f_m\}$, chosen with probabilities 
$\mathbf{p}^{i}$ \;
\If{$x^i(st +h) = 1$}{
Car $i$ {\bf scans} for missing entity using RFID \; 
\If{the missing entity is located}{
Car $i$ {\bf returns} position of the missing entity to the requester of the search \;
}
}
}
Central authority {\bf observes} filtered aggregate state $\hat y(st +h)$, where the filter ${\cal F}$ is possibly not known a priori \;
Central authority {\bf computes} the error $e(st +h)$ \;
Central authority {\bf broadcasts} signal $\pi(st +h)$ computed using some controller ${\cal C}$ and increments $h$ to $h +1$\;
}
The environment {\bf perturbs} the numbers $N_i$, i.e., the numbers of cars parked in the vicinity, such that $| N^i((s+1)t) - N^i(st + h) | \le \delta $ and increments $s$ to $s + 1$. 
}
The cars have {\bf failed} to locate the entity within the time horizon \;
{\bf Return} an alert to the requester of the search\;
\label{algo01}
\end{algorithm}
To demonstrate the performance of our algorithm, we employed Simulation of Urban MObility (SUMO) Version 1.2.0.
SUMO \cite{Krajzewicz2012} is an open-source, microscopic traffic simulation package primarily being developed at the Institute of Transportation Systems at the German Aerospace Centre (DLR).
SUMO is designed to handle large networks and comes with a ``remote control'' interface,
TraCI (short for Traffic Control Interface) \cite{Wegeneretal2008}, which allows one to adapt the
simulation and to control singular vehicles and pedestrians on the fly.
Our goal was to simulate a pedestrian walking about in an urban scenario, and to regulate the number
of parked vehicles actively searching for the pedestrian in an energy- and coverage-efficient manner using our algorithm.

\subsection{City of Melbourne test case scenario setup}
The region considered for our simulations consisted of the City of Melbourne municipality, with a boundary map obtained from \cite{CityOfMelbourneBoundaryMap}; cf. Fig. \ref{mapofregion}. A dataset containing spatial polygons representing the on-street parking bays across the city was obtained from \cite{ParkingBaysMap}. A total of 24,067 on-street parking spaces were imported to our SUMO network as polygons from this dataset.

\begin{figure}[th]
\centering
\includegraphics[width=0.95\columnwidth]{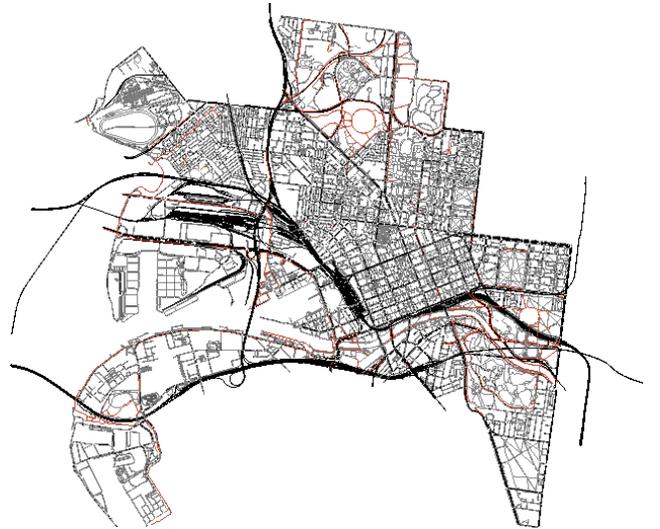}
\caption{A map of the City of Melbourne, as imported from OpenStreetMap for our use in SUMO simulations.}\label{mapofregion}
\end{figure}

To generate random walks for the pedestrian, we utilised the TraCI function \emph{traci.simulation.findIntermodalRoute}. In particular, at the commencement of each walk, a random origin and destination lane were selected from the list of all possible lanes on the network for which pedestrians were permitted on, and these origin and destination links were then provided as input to the TraCI function which generated the route. The maximum walking speed for the pedestrian was set at SUMO's default of 1.39m/s. We used SUMO's \emph{striping} pedestrian model \cite{SUMO-striping} as the model for how the person otherwise interacted with the map.

Another parameter in our experiment was the proportion of parking spaces in each simulation that would have cars parked in them that were capable of participating in the search. We elected for each parking space (out of the 24,067 total parking spaces) to have a 50\% chance of being inhabited by a vehicle capable of participating in the search.  Thus, at the beginning of each simulation, a ``coin'' was flipped for each of the 24,067 parking spaces. The result of this ``coin flip'' was compared to the fifty percent value to determine whether that parking space would be inhabited by a parked vehicle capable of participating in the search or not, over that particular simulation.  Parking space assignments for vehicles then remained constant for the duration of a simulation, and parked, participating vehicles were ``Switched On'' or ``Switched Off'' according to Algorithm \ref{algo01}. At the beginning of each search, the proportion of participating vehicles that were initially ``Switched On'' was set at 30\%. We chose our target number $r$ of ``Switched On'' vehicles to be 7,200.

For our probability models, we employed the use of logistic functions which are illustrated in Fig. \ref{probabilitycurves_melbourne}. We placed a circle with a radius of twenty metres around each parked vehicle capable of participating in the search, and let the number of other parked vehicles (capable of participating in the search, and) residing within this circle, equate to the number of neighbours that the vehicle at the centre of the circle had.  For simplicity, we assume that $\hat y = y$ (that is, the filter ${\cal F}$ provides a perfect estimate for the resource consumption)\footnote{Moving average schemes are also standard choices.}. We also consider a simple controller model given by the difference equation
\begin{equation}\label{eq:simpl-cntrlr}
\pi(k) = \gamma \pi(k - 1) + \kappa \left[e(k) - \alpha e(k-1) \right],
\end{equation}
for all $k \in \N$, in which $\alpha,\gamma,\kappa \in \R$. This model includes, as particular cases, classical lead,
lag and PI controller structures \cite{Franklin,Franklin_dig}. For this particular example, 
 we let $\alpha=-4.01$, $\gamma=0.99$ and $\kappa=0.1$ in \eqref{eq:simpl-cntrlr}.  We set each vehicle's RFID polling rate (i.e., the frequency at which a car's RFID system is sampling when the vehicle is ``Switched On'') as ``Always On'', meaning that once ``Switched On'', a vehicle is always polling as opposed to doing periodic, timed reads. We set a circular RFID field around each car with a radius of six metres. Moreover, we assumed that once a pedestrian entered this field, and if the vehicle was ``Switched On'', then the pedestrian would be detected.  In other words, in this paper, we neglect some of the more complicated phenomena typically associated with RFID, such as the effects of tag placement, antenna orientation, cable length, reader settings, and environmental factors such as the existence of water or other radio waves \cite{RFIDInsider}.

\begin{figure}[t!]
\centering
\includegraphics[width=0.9\columnwidth]{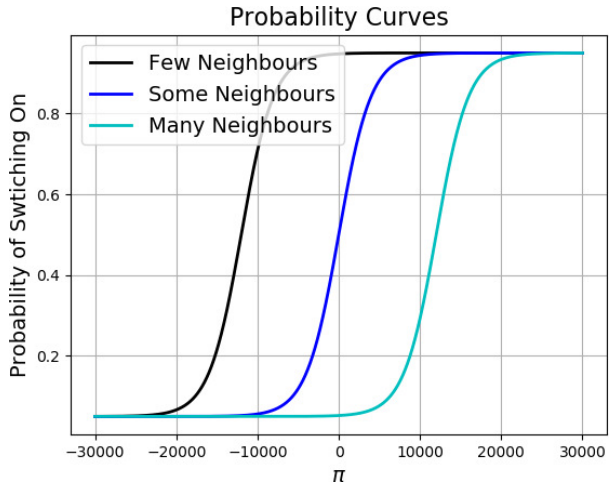}
\caption{Logistic functions used to model the possible behaviours  $\{f_f, f_s, f_m\}$ in Algorithm~\ref{algo01}.}\label{probabilitycurves_melbourne}
\end{figure}

For each simulation, then, our goal was to set the person down on a random edge, and have them walk until either: (i) they were detected by a parked vehicle that was ``Switched On'' and thus actively searching at the same time as when the pedestrian was passing by; or (ii) thirty minutes had transpired and no detection event had occurred. We permitted thirty minutes to lapse before a ``fail-to-detect'' event was recorded, keeping in mind that quickly finding a missing and potentially stressed person, and returning them to their home, for instance, is ideal. All simulations had time-step updates of 1s, while our control signals were sent only every 20s. For our test case scenario, 100 simulations were performed in total.

\subsection{Numerical illustrations}\label{sec:results}
To gather some preliminary data, we first permitted a small sample of ten simulations to run for a full thirty minutes, with no pedestrian placement yet.  From these simulations, Fig. \ref{updatesevery20s} demonstrates that regulation of the system, such that approximately 7,200 parked vehicles were ``Switched On'' at any point in time, was achieved quite rapidly. Specifically, the blue line in the figure indicates the mean number of vehicles "Switched On' versus time (from the ten sample simulations); while the red shaded area indicates one sample standard deviation each side of the mean. Fig. \ref{controlsignalversustime} illustrates the evolution of the mean control signal $\pi$ over time. (Again, the red shaded area indicates one sample standard deviation each side of the mean.) Notice that $\pi$ could then be used in association with Fig. \ref{probabilitycurves_melbourne}, along with the known number of neighbours that a vehicle had, to determine the probability of that vehicle being ``Switched On'' over the next appropriate time interval.

\begin{figure}[t!]
\centering
\includegraphics[width=0.95\columnwidth]{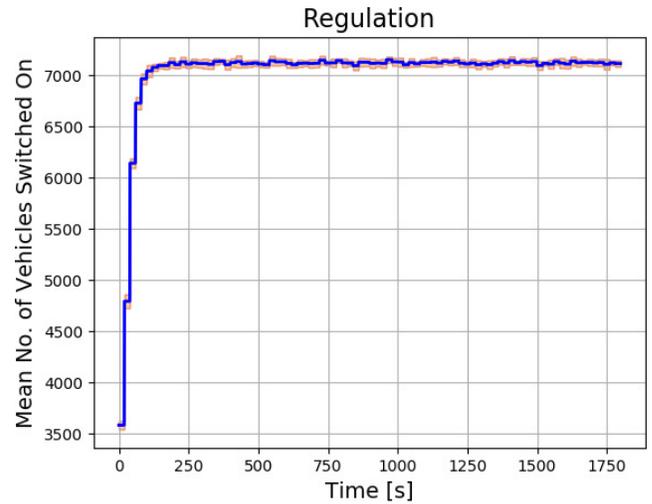}
\caption{The blue line indicates the mean number of vehicles "Switched On' versus time (from ten sample simulations, each regulating the number of ``Switched On'' vehicles to 7,200 at any point in time), while the red shaded area indicates the area within one standard deviation from the mean number of vehicles.}\label{updatesevery20s}
\end{figure}

\begin{figure}[t!]
\centering
\includegraphics[width=0.95\columnwidth]{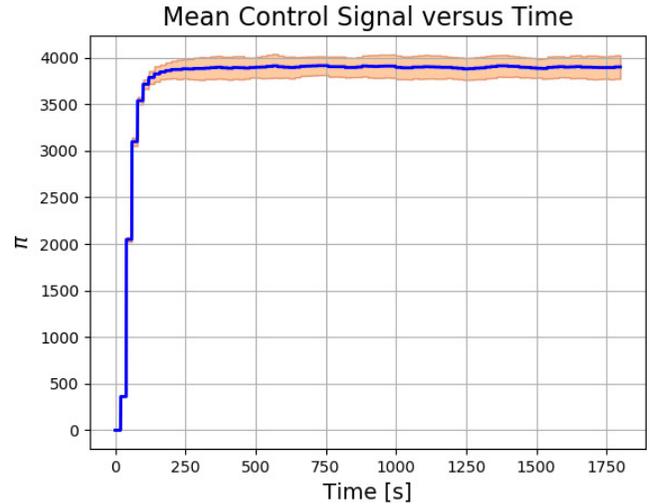}
\caption{Weak convergence of the mean control signal, $\pi$, over time. The red shaded area indicates the area within one  standard deviation away from the mean control signal.}\label{controlsignalversustime}
\end{figure}

Next, we performed our simulations proper, where a pedestrian was inserted onto the map at the beginning of each simulation, and these ran until either: (i) the pedestrian was detected by a parked vehicle that was ``Switched On'' and thus actively searching at the same time as when the pedestrian was passing by; or (ii) thirty minutes had lapsed and no detection event had occurred.  The data collected from our experiment comprised of: (i) the average time taken (in minutes) until detection of the missing entity occurred (provided that the detection occurred within thirty minutes from the beginning of an emulation, else a fail result was recorded); and (ii) the total number of times that fail results were recorded over the entirety of the experiment.  To reiterate, 100 simulations in total were conducted during our experiment. The results were as follows: (a) Average Detection Time $=$ 5.30 minutes; and (b) Failed to Detect $=$ 6 times out of 100 simulations. In other words, the pedestrian was not detected within a thirty-minute time frame, 6\% of the time.  For the other 94 cases, the pedestrian was detected, on average, in approximately five minutes.

\section{CONCLUSIONS AND FUTURE WORK}\label{conclusions}
We have considered the notion of predictability and 
a notion of fairness in time-varying probabilistic models of social sensing.
This could be seen as a contribution to the growing literature \cite{Fioravanti18,narasimhan2020pairwise,zhou2020fairness} on fairness beyond machine learning, as well as an addition to the theory of social sensing.

A number of theoretical questions arise: what other conditions assure fairness in the
sense of statistical parity (Definition \ref{dfn:fair})? What other notions of fairness could there be, other than Definitions~\ref{dfn:fair}--\ref{dfn:eps-fair}?
We believe these could spur a considerable interest across both Social Sensing
and Control Theory.

In our application, we have considered dynamic parking, which requires such time-varying 
probabilistic models.
We envisage a number of ways forward regarding improving our experimental setup, 
including performing more simulations in further cities worldwide.

There could also be a number of other applications. For instance, during the current 
COVID-19 pandemic \cite{rashid2020covidsens}, many governments considered the mandatory participation in a tracing 
scheme that would be sufficient to contain a contagion, and the option of 
invading the privacy of individuals in such a sensing scheme. 
One could also see testing as a means of social sensing and consider a stochastic model \cite{marecek2020screening} thereof.
In such a setting, our notion of fairness may also be worth considering.

\small 

\vspace{1\baselineskip}
\par\noindent
\parbox[t]{\linewidth}{
\noindent\parpic{\includegraphics[height=2 in,width=1 in,clip,keepaspectratio]{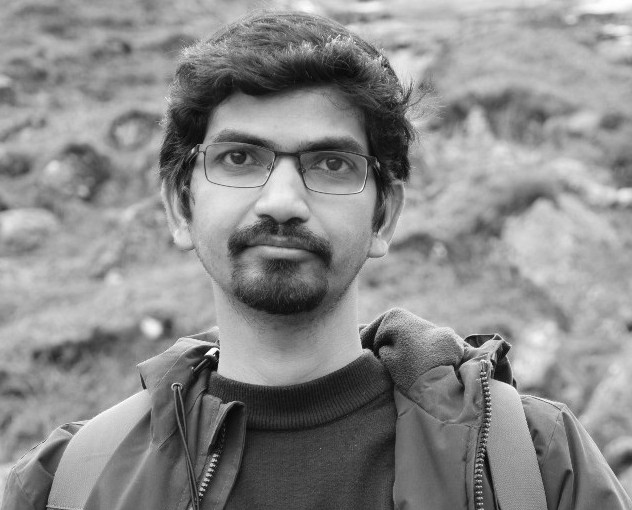}}
\noindent {\bf Ramen Ghosh}\
received his Bachelor of Science in Mathematics (Honours) at the University of Calcutta, and his Master of Science in Mathematics at Chennai Mathematical Institute, Chennai, India, and his Master of Technology in Mathematics and Computing at the Indian Institute of Technology, Patna, India, respectively in 2009, 2011 and 2017. He is currently a Ph.D. student in Control Engineering and Decision Science at the School of Electrical and Electronic Engineering at University College Dublin, Ireland. His current research interests include iterated function systems, stochastic processes, and dynamical systems arising in control theory.
}
\vspace{1\baselineskip}
\par\noindent
\parbox[t]{\linewidth}{
\noindent\parpic{\includegraphics[height=1.5in,width=1 in,clip,keepaspectratio]{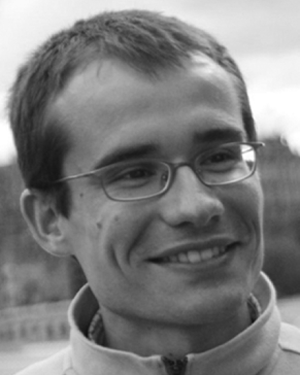}}
\noindent {\bf Jakub Mare{\v c}ek}\
received his first two degrees from Masaryk University in Brno, the Czech Republic, and his Ph.D.
degree from the University of Nottingham, Nottingham, U.K., in 2006, 2009, and 2012, respectively.
He has worked in two start-ups, at ARM Ltd., at the University of Edinburgh, at the University of Toronto, 
at IBM Research -- Ireland, 
and at the University of California, Los Angeles.
He is currently a faculty member at the Czech Technical University in Prague, the Czech Republic.
He designs and analyses algorithms for optimisation and control problems across a range of application domains,
including power systems, transportation, and robust statistics.
}
\vspace{4\baselineskip}
\par\noindent
\parbox[t]{\linewidth}{
\noindent\parpic{\includegraphics[height=1.5in,width=1 in,clip,keepaspectratio]{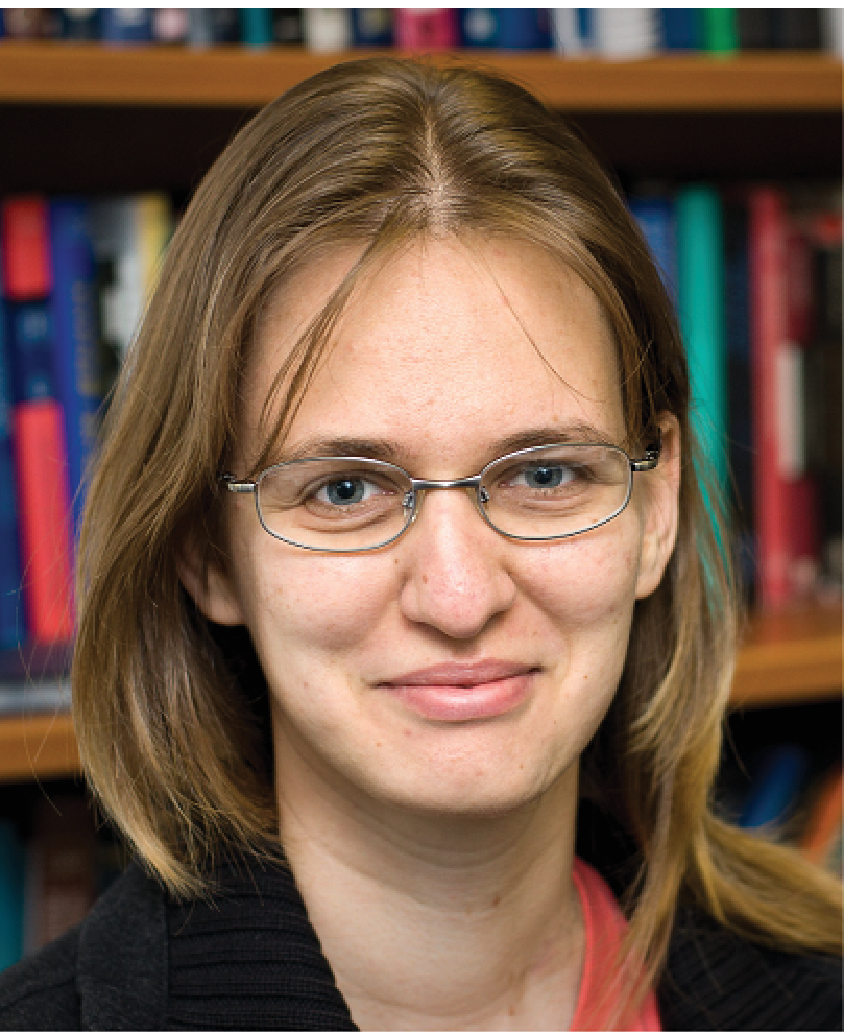}}
\noindent {\bf Wynita M. Griggs}\
received her Ph.D. degree in Engineering from the Australian National University in Canberra, Australia, in 2007. Between 2008 and 2015, she was a Postdoctoral Research Fellow at the Hamilton Institute, National University of Ireland Maynooth in Maynooth, Ireland. From 2015 to 2019, she was a Research Scientist at University College Dublin in Dublin, Ireland. She is currently a Lecturer at Monash University in Melbourne, Australia. Her research interests include stability theory with applications to feedback control systems; and intelligent transportation systems.
}
\vspace{1\baselineskip}
\par\noindent
\parbox[t]{\linewidth}{
\noindent\parpic{\includegraphics[height=1.5in,width=1in,clip,keepaspectratio]{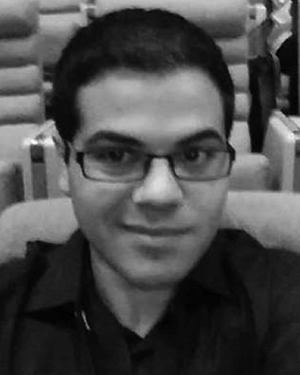}}
\noindent {\bf Matheus Souza}
obtained his BEng, MSc, Ph.D. degrees from the School of Electrical and Computer Engineering (FEEC), University of Campinas (UNICAMP).
He visited Maynooth University as a PhD visiting researcher and he worked as a Post Doctoral Research Fellow at University College Dublin.
He is now an Assistant Professor at UNICAMP and his current research interests include analysis and design of dynamical systems and mathematical optimisation with applications to smart cities.
}
\vspace{1\baselineskip}
\par\noindent
\parbox[t]{\linewidth}{
\noindent\parpic{\includegraphics[height=1.5in,width=1 in,clip,keepaspectratio]{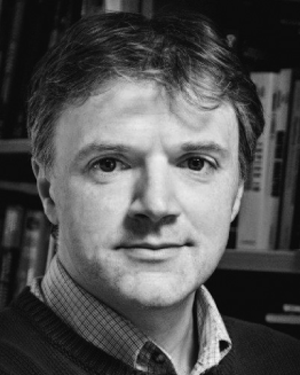}}
\noindent {\bf Robert N. Shorten}\
received his Ph.D. degree from University College Dublin, Ireland in 1996.
From 1993 to 1996, he was the Holder of a Marie Curie Fellowship with Daimler-Benz Berlin, Germany,
where he conducted research in the area of smart gearbox systems. Following a brief period with the Centre for Systems Science,
Yale University, New Haven, CT, USA, where he was involved with research with Prof. K. S. Narendra, he returned to Ireland as the Holder of a
European Presidency Fellowship in 1997. He is the co-founder of the Hamilton Institute, National University of Ireland Maynooth, Ireland,
where he was a Full Professor until 2013. From 2013 to 2015, he was a Senior Research Manager with IBM Research -- Ireland, where he led
the control and optimization activities. He currently holds appointments at Imperial College London and University College Dublin,
where he is a Professor of control engineering and decision science.
}

\bibliographystyle{{IEEEtranN}}
\bibliography{ref}

\begin{thebibliography}{85}
\providecommand{\natexlab}[1]{#1}
\providecommand{\url}[1]{#1}
\csname url@samestyle\endcsname
\providecommand{\newblock}{\relax}
\providecommand{\bibinfo}[2]{#2}
\providecommand{\BIBentrySTDinterwordspacing}{\spaceskip=0pt\relax}
\providecommand{\BIBentryALTinterwordstretchfactor}{4}
\providecommand{\BIBentryALTinterwordspacing}{\spaceskip=\fontdimen2\font plus
\BIBentryALTinterwordstretchfactor\fontdimen3\font minus
  \fontdimen4\font\relax}
\providecommand{\BIBforeignlanguage}[2]{{%
\expandafter\ifx\csname l@#1\endcsname\relax
\typeout{** WARNING: IEEEtranN.bst: No hyphenation pattern has been}%
\typeout{** loaded for the language `#1'. Using the pattern for}%
\typeout{** the default language instead.}%
\else
\language=\csname l@#1\endcsname
\fi
#2}}
\providecommand{\BIBdecl}{\relax}
\BIBdecl

\bibitem[Souza et~al.(2018)Souza, Griggs, Mare\v{c}ek, and
  Shorten]{souza_conference}
M.~Souza, W.~M. Griggs, J.~Mare\v{c}ek, and R.~N. Shorten, ``Regulating the
  searching behaviour of parked vehicles attempting to locate moving, missing
  entities,'' in \emph{Proceedings of the 21st International IEEE Conference on
  Intelligent Transportation Systems}, 2018.

\bibitem[Howe(2006)]{howe2006rise}
J.~Howe, ``The rise of crowdsourcing,'' \emph{Wired magazine}, vol.~14, no.~6,
  pp. 1--4, 2006.

\bibitem[Lu et~al.(2014)Lu, Cheng, Zhang, Shen, and Mark]{Luetal2014}
N.~Lu, N.~Cheng, N.~Zhang, X.~Shen, and J.~W. Mark, ``Connected vehicles:
  solutions and challenges,'' \emph{IEEE Internet of Things Journal}, vol. 1,
  no. 4, pp. 289--299, 2014.

\bibitem[Zanella et~al.(2014)Zanella, Bui, Castellani, Vangelista, and
  Zorzi]{6740844}
A.~Zanella, N.~Bui, A.~Castellani, L.~Vangelista, and M.~Zorzi, ``Internet of
  things for smart cities,'' \emph{Internet of Things Journal, IEEE}, vol.~1,
  no.~1, pp. 22--32, Feb 2014.

\bibitem[Artikis et~al.(2014)Artikis, Weidlich, Schnitzler, Boutsis, Liebig,
  Piatkowski, Bockermann, Morik, Kalogeraki, Marecek,
  et~al.]{artikis2014heterogeneous}
A.~Artikis, M.~Weidlich, F.~Schnitzler, I.~Boutsis, T.~Liebig, N.~Piatkowski,
  C.~Bockermann, K.~Morik, V.~Kalogeraki, J.~Marecek \emph{et~al.},
  ``Heterogeneous stream processing and crowdsourcing for urban traffic
  management.'' in \emph{EDBT}, vol.~14, 2014, pp. 712--723.

\bibitem[{Cogill} et~al.(2014){Cogill}, {Gallay}, {Griggs}, {Lee}, {Nabi},
  {Ordonez}, {Rufli}, {Shorten}, {Tchrakian}, {Verago}, {Wirth}, and
  {Zhuk}]{Cogilletal2014}
R.~{Cogill}, O.~{Gallay}, W.~{Griggs}, C.~{Lee}, Z.~{Nabi}, R.~{Ordonez},
  M.~{Rufli}, R.~{Shorten}, T.~{Tchrakian}, R.~{Verago}, F.~{Wirth}, and
  S.~{Zhuk}, ``Parked cars as a service delivery platform,'' in \emph{2014
  International Conference on Connected Vehicles and Expo (ICCVE)}, 2014, pp.
  138--143.

\bibitem[{Wang} et~al.(2013){Wang}, {Kaplan}, {Abdelzaher}, and
  {Aggarwal}]{6517107}
D.~{Wang}, L.~{Kaplan}, T.~{Abdelzaher}, and C.~C. {Aggarwal}, ``On credibility
  estimation tradeoffs in assured social sensing,'' \emph{IEEE Journal on
  Selected Areas in Communications}, vol.~31, no.~6, pp. 1026--1037, 2013.

\bibitem[{Meyer} et~al.(2016){Meyer}, {Hlinka}, {Wymeersch}, {Riegler}, and
  {Hlawatsch}]{7365472}
F.~{Meyer}, O.~{Hlinka}, H.~{Wymeersch}, E.~{Riegler}, and F.~{Hlawatsch},
  ``Distributed localization and tracking of mobile networks including
  noncooperative objects,'' \emph{IEEE Transactions on Signal and Information
  Processing over Networks}, vol.~2, no.~1, pp. 57--71, 2016.

\bibitem[{Perera} et~al.(2015){Perera}, {Talagala}, {Liu}, and
  {Estrella}]{7397993}
C.~{Perera}, D.~S. {Talagala}, C.~H. {Liu}, and J.~C. {Estrella},
  ``Energy-efficient location and activity-aware on-demand mobile distributed
  sensing platform for sensing as a service in iot clouds,'' \emph{IEEE
  Transactions on Computational Social Systems}, vol.~2, no.~4, pp. 171--181,
  2015.

\bibitem[{Huang} and {Wang}(2017)]{7932851}
C.~{Huang} and D.~{Wang}, ``Exploring scalability and time-sensitiveness in
  reliable social sensing with accuracy assessment,'' \emph{IEEE Access},
  vol.~5, pp. 14\,405--14\,418, 2017.

\bibitem[{Cepni} et~al.(2018){Cepni}, {Ozger}, and {Akan}]{8382177}
K.~{Cepni}, M.~{Ozger}, and O.~B. {Akan}, ``Event estimation accuracy of social
  sensing with facebook for social internet of vehicles,'' \emph{IEEE Internet
  of Things Journal}, vol.~5, no.~4, pp. 2449--2456, 2018.

\bibitem[{Zhang} et~al.(2020){Zhang}, {Ma}, {Hu}, and {Wang}]{9105079}
D.~{Zhang}, Y.~{Ma}, X.~S. {Hu}, and D.~{Wang}, ``Towards privacy-aware task
  allocation in social sensing based edge computing systems,'' \emph{IEEE
  Internet of Things Journal}, pp. 1--1, 2020.

\bibitem[To et~al.(2014)To, Ghinita, and Shahabi]{To2014}
H.~To, G.~Ghinita, and C.~Shahabi, ``A framework for protecting worker location
  privacy in spatial crowdsourcing,'' \emph{Proc. VLDB Endow.}, vol.~7, no.~10,
  p. 919–930, Jun. 2014.

\bibitem[{Guo} et~al.(2018){Guo}, {Liu}, {Wang}, {Li}, {Lam}, and
  {Yu}]{8316812}
B.~{Guo}, Y.~{Liu}, L.~{Wang}, V.~O.~K. {Li}, J.~C.~K. {Lam}, and Z.~{Yu},
  ``Task allocation in spatial crowdsourcing: Current state and future
  directions,'' \emph{IEEE Internet of Things Journal}, vol.~5, no.~3, pp.
  1749--1764, 2018.

\bibitem[{Wang} et~al.(2018){Wang}, {Wang}, {Wang}, {Zhang}, and
  {Kong}]{8429062}
J.~{Wang}, L.~{Wang}, Y.~{Wang}, D.~{Zhang}, and L.~{Kong}, ``Task allocation
  in mobile crowd sensing: State-of-the-art and future opportunities,''
  \emph{IEEE Internet of Things Journal}, vol.~5, no.~5, pp. 3747--3757, 2018.

\bibitem[Guo et~al.(2015)Guo, Wang, Yu, Wang, Yen, Huang, and
  Zhou]{guo2015mobile}
B.~Guo, Z.~Wang, Z.~Yu, Y.~Wang, N.~Y. Yen, R.~Huang, and X.~Zhou, ``Mobile
  crowd sensing and computing: The review of an emerging human-powered sensing
  paradigm,'' \emph{ACM Computing Surveys (CSUR)}, vol.~48, no.~1, pp. 1--31,
  2015.

\bibitem[{Wang} et~al.(2019){Wang}, {Szymanski}, {Abdelzaher}, {Ji}, and
  {Kaplan}]{8666667}
D.~{Wang}, B.~K. {Szymanski}, T.~{Abdelzaher}, H.~{Ji}, and L.~{Kaplan}, ``The
  age of social sensing,'' \emph{Computer}, vol.~52, no.~1, pp. 36--45, 2019.

\bibitem[Wang et~al.(2015)Wang, Abdelzaher, and Kaplan]{wang2015social}
D.~Wang, T.~Abdelzaher, and L.~Kaplan, \emph{Social sensing: building reliable
  systems on unreliable data}.\hskip 1em plus 0.5em minus 0.4em\relax Morgan
  Kaufmann, 2015.

\bibitem[Stolton(2021)]{Euractive2021}
S.~Stolton, ``Von der leyen assures meps: We’ll ‘go further’ on ai that
  harms fundamental rights,'' \emph{Euractive.com}, no. March 30th, 2021.

\bibitem[Santos and Peralta-Alva(2005)]{santos2005accuracy}
M.~S. Santos and A.~Peralta-Alva, ``Accuracy of simulations for stochastic
  dynamic models,'' \emph{Econometrica}, vol.~73, no.~6, pp. 1939--1976, 2005.

\bibitem[Ope({\natexlab{a}})]{OpenClipArt1}
``https://openclipart.org/detail/234440 accessed: Apr. 16, 2018.''

\bibitem[Ope({\natexlab{b}})]{OpenClipArt2}
``https://openclipart.org/detail/201733 accessed: Apr. 16, 2018.''

\bibitem[{Alzheimer's Association}(Accessed: Sep. 4,
  2018)]{AlzheimersAssociationIII}
\BIBentryALTinterwordspacing
{Alzheimer's Association}, ``Medicalert + safe return,'' Accessed: Sep. 4,
  2018. [Online]. Available:
  \url{https://www.alz.org/help-support/caregiving/safety/medicalert-safe-return}
\BIBentrySTDinterwordspacing

\bibitem[{Verago} et~al.(2015){Verago}, {Naoum-Sawaya}, {Griggs}, and
  {Shorten}]{Veragoetal2015}
R.~{Verago}, J.~{Naoum-Sawaya}, W.~{Griggs}, and R.~{Shorten}, ``Localization
  of missing entities using parked cars,'' in \emph{2015 International
  Conference on Connected Vehicles and Expo (ICCVE)}.\hskip 1em plus 0.5em
  minus 0.4em\relax IEEE, 2015, pp. 40--41.

\bibitem[{Griggs} et~al.(2018){Griggs}, {Verago}, {Naoum-Sawaya},
  {Ordóñez-Hurtado}, {Gilmore}, and {Shorten}]{Griggsetal2018}
W.~M. {Griggs}, R.~{Verago}, J.~{Naoum-Sawaya}, R.~H. {Ordóñez-Hurtado},
  R.~{Gilmore}, and R.~N. {Shorten}, ``Localizing missing entities using parked
  vehicles: An rfid-based system,'' \emph{IEEE Internet of Things Journal},
  vol.~5, no.~5, pp. 4018--4030, 2018.

\bibitem[Landaluce et~al.(2020)Landaluce, Arjona, Perallos, Falcone, Angulo,
  and Muralter]{landaluce2020review}
H.~Landaluce, L.~Arjona, A.~Perallos, F.~Falcone, I.~Angulo, and F.~Muralter,
  ``A review of iot sensing applications and challenges using rfid and wireless
  sensor networks,'' \emph{Sensors}, vol.~20, no.~9, p. 2495, 2020.

\bibitem[Fioravanti et~al.(2019)Fioravanti, Mareček, Shorten, Souza, and
  Wirth]{Fioravanti18}
A.~R. Fioravanti, J.~Mareček, R.~N. Shorten, M.~Souza, and F.~R. Wirth, ``On
  the ergodic control of ensembles,'' \emph{Automatica}, vol. 108, p. 108483,
  2019.

\bibitem[Van Der~Vaart and Wellner(1996)]{van1996weak}
A.~W. Van Der~Vaart and J.~A. Wellner, \emph{Weak convergence and empirical
  processes}.\hskip 1em plus 0.5em minus 0.4em\relax Springer, 1996.

\bibitem[Restuccia et~al.(2017)Restuccia, Ghosh, Bhattacharjee, Das, and
  Melodia]{restuccia2017quality}
F.~Restuccia, N.~Ghosh, S.~Bhattacharjee, S.~K. Das, and T.~Melodia, ``Quality
  of information in mobile crowdsensing: Survey and research challenges,''
  \emph{ACM Transactions on Sensor Networks (TOSN)}, vol.~13, no.~4, pp. 1--43,
  2017.

\bibitem[{Marano} et~al.(2016){Marano}, {Matta}, and {Willett}]{7456345}
S.~{Marano}, V.~{Matta}, and P.~{Willett}, ``The importance of being earnest:
  Social sensing with unknown agent quality,'' \emph{IEEE Transactions on
  Signal and Information Processing over Networks}, vol.~2, no.~3, pp.
  306--320, 2016.

\bibitem[{Wu} et~al.(2014){Wu}, {Wang}, {Wang}, {Zhang}, and {Tung}]{6529080}
S.~{Wu}, X.~{Wang}, S.~{Wang}, Z.~{Zhang}, and A.~K.~H. {Tung}, ``K-anonymity
  for crowdsourcing database,'' \emph{IEEE Transactions on Knowledge and Data
  Engineering}, vol.~26, no.~9, pp. 2207--2221, 2014.

\bibitem[{To} et~al.(2017){To}, {Ghinita}, {Fan}, and {Shahabi}]{7501846}
H.~{To}, G.~{Ghinita}, L.~{Fan}, and C.~{Shahabi}, ``Differentially private
  location protection for worker datasets in spatial crowdsourcing,''
  \emph{IEEE Transactions on Mobile Computing}, vol.~16, no.~4, pp. 934--949,
  2017.

\bibitem[Wang et~al.(2016)Wang, Cai, Yin, Gao, Tong, and Wu]{WANG2016}
Y.~Wang, Z.~Cai, G.~Yin, Y.~Gao, X.~Tong, and G.~Wu, ``An incentive mechanism
  with privacy protection in mobile crowdsourcing systems,'' \emph{Computer
  Networks}, vol. 102, pp. 157 -- 171, 2016.

\bibitem[{Chi} et~al.(2018){Chi}, {Wang}, {Huang}, and {Tong}]{8207350}
Z.~{Chi}, Y.~{Wang}, Y.~{Huang}, and X.~{Tong}, ``The novel location
  privacy-preserving ckd for mobile crowdsourcing systems,'' \emph{IEEE
  Access}, vol.~6, pp. 5678--5687, 2018.

\bibitem[{Shu} et~al.(2018){Shu}, {Liu}, {Jia}, {Yang}, and {Deng}]{8351913}
J.~{Shu}, X.~{Liu}, X.~{Jia}, K.~{Yang}, and R.~H. {Deng}, ``Anonymous
  privacy-preserving task matching in crowdsourcing,'' \emph{IEEE Internet of
  Things Journal}, vol.~5, no.~4, pp. 3068--3078, 2018.

\bibitem[{Wang} et~al.(2020){Wang}, {Gu}, {Ma}, and {Jin}]{8909376}
Y.~{Wang}, M.~{Gu}, J.~{Ma}, and Q.~{Jin}, ``Dnn-dp: Differential privacy
  enabled deep neural network learning framework for sensitive crowdsourcing
  data,'' \emph{IEEE Transactions on Computational Social Systems}, vol.~7,
  no.~1, pp. 215--224, 2020.

\bibitem[{Feng} et~al.(2018){Feng}, {Yan}, {Zhang}, {Zeng}, {Xiao}, and
  {Hou}]{8080202}
W.~{Feng}, Z.~{Yan}, H.~{Zhang}, K.~{Zeng}, Y.~{Xiao}, and Y.~T. {Hou}, ``A
  survey on security, privacy, and trust in mobile crowdsourcing,'' \emph{IEEE
  Internet of Things Journal}, vol.~5, no.~4, pp. 2971--2992, 2018.

\bibitem[{Jiang} et~al.(2019){Jiang}, {Fang}, and {Wang}]{8482313}
T.~{Jiang}, H.~{Fang}, and H.~{Wang}, ``Blockchain-based internet of vehicles:
  Distributed network architecture and performance analysis,'' \emph{IEEE
  Internet of Things Journal}, vol.~6, no.~3, pp. 4640--4649, 2019.

\bibitem[{Zou} et~al.(2020){Zou}, {Xi}, {Wang}, and {Xu}]{8926541}
S.~{Zou}, J.~{Xi}, H.~{Wang}, and G.~{Xu}, ``Crowdblps: A blockchain-based
  location-privacy-preserving mobile crowdsensing system,'' \emph{IEEE
  Transactions on Industrial Informatics}, vol.~16, no.~6, pp. 4206--4218,
  2020.

\bibitem[Eriksson et~al.(2008)Eriksson, Girod, Hull, Newton, Madden, and
  Balakrishnan]{eriksson2008pothole}
J.~Eriksson, L.~Girod, B.~Hull, R.~Newton, S.~Madden, and H.~Balakrishnan,
  ``The pothole patrol: using a mobile sensor network for road surface
  monitoring,'' in \emph{Proceedings of the 6th international conference on
  Mobile systems, applications, and services}, 2008, pp. 29--39.

\bibitem[Mohan et~al.(2008)Mohan, Padmanabhan, and Ramjee]{mohan2008nericell}
P.~Mohan, V.~N. Padmanabhan, and R.~Ramjee, ``Nericell: rich monitoring of road
  and traffic conditions using mobile smartphones,'' in \emph{Proceedings of
  the 6th ACM conference on Embedded network sensor systems}, 2008, pp.
  323--336.

\bibitem[Jeske(2013)]{jeske2013floating}
T.~Jeske, ``Floating car data from smartphones: What google and waze know about
  you and how hackers can control traffic,'' \emph{Proc. of the BlackHat
  Europe}, pp. 1--12, 2013.

\bibitem[Stevens and D’Hondt(2010)]{stevens2010crowdsourcing}
M.~Stevens and E.~D’Hondt, ``Crowdsourcing of pollution data using
  smartphones,'' in \emph{Workshop on Ubiquitous Crowdsourcing}, 2010, pp.
  1--4.

\bibitem[Rana et~al.(2010)Rana, Chou, Kanhere, Bulusu, and Hu]{rana2010ear}
R.~K. Rana, C.~T. Chou, S.~S. Kanhere, N.~Bulusu, and W.~Hu, ``Ear-phone: an
  end-to-end participatory urban noise mapping system,'' in \emph{Proceedings
  of the 9th ACM/IEEE international conference on information processing in
  sensor networks}, 2010, pp. 105--116.

\bibitem[Maisonneuve et~al.(2010)Maisonneuve, Stevens, and
  Ochab]{maisonneuve2010participatory}
N.~Maisonneuve, M.~Stevens, and B.~Ochab, ``Participatory noise pollution
  monitoring using mobile phones,'' \emph{Information polity}, vol.~15, no. 1,
  2, pp. 51--71, 2010.

\bibitem[{Kun-Ying Lin} et~al.(2011){Kun-Ying Lin}, {Ming-Wei Hsu}, and
  {Shi-Rung Liou}]{6037138}
{Kun-Ying Lin}, {Ming-Wei Hsu}, and {Shi-Rung Liou}, ``Bicycle management
  systems in anti-theft, certification, and race by using rfid,'' in
  \emph{Proceedings of 2011 Cross Strait Quad-Regional Radio Science and
  Wireless Technology Conference}, vol.~2, 2011, pp. 1054--1057.

\bibitem[{Liu} and {Li}(2014)]{7035659}
K.~{Liu} and X.~{Li}, ``Finding nemo: Finding your lost child in crowds via
  mobile crowd sensing,'' in \emph{2014 IEEE 11th International Conference on
  Mobile Ad Hoc and Sensor Systems}, 2014, pp. 1--9.

\bibitem[{El Alaoui El Abdallaoui} et~al.(2016){El Alaoui El Abdallaoui},
  {Abdelaziz}, and {Mohamed}]{7929047}
H.~{El Alaoui El Abdallaoui}, E.~F. {Abdelaziz}, and S.~{Mohamed}, ``Finding a
  lost child using a crowdsourcing framework,'' in \emph{2016 4th International
  Conference on Control Engineering Information Technology (CEIT)}, 2016, pp.
  1--6.

\bibitem[Rashid and Wang(2020)]{rashid2020covidsens}
M.~T. Rashid and D.~Wang, ``Covidsens: a vision on reliable social sensing for
  covid-19,'' \emph{Artificial Intelligence Review}, pp. 1--25, 2020.

\bibitem[Chang et~al.(2020)Chang, Pierson, Koh, Gerardin, Redbird, Grusky, and
  Leskovec]{Chang2020}
\BIBentryALTinterwordspacing
S.~Chang, E.~Pierson, P.~W. Koh, J.~Gerardin, B.~Redbird, D.~Grusky, and
  J.~Leskovec, ``Mobility network models of covid-19 explain inequities and
  inform reopening,'' \emph{Nature}, 2020. [Online]. Available:
  \url{https://doi.org/10.1038/s41586-020-2923-3}
\BIBentrySTDinterwordspacing

\bibitem[Chatzimilioudis et~al.(2012)Chatzimilioudis, Konstantinidis, Laoudias,
  and Zeinalipour-Yazti]{chatzimilioudis2012crowdsourcing}
G.~Chatzimilioudis, A.~Konstantinidis, C.~Laoudias, and D.~Zeinalipour-Yazti,
  ``Crowdsourcing with smartphones,'' \emph{IEEE Internet Computing}, vol.~16,
  no.~5, pp. 36--44, 2012.

\bibitem[Mare{\v{c}}ek et~al.(2015)Mare{\v{c}}ek, Shorten, and
  Yu]{marevcek2015signalling}
J.~Mare{\v{c}}ek, R.~Shorten, and J.~Y. Yu, ``Signalling and obfuscation for
  congestion control,'' \emph{International Journal of Control}, vol.~88,
  no.~10, pp. 2086--2096, 2015.

\bibitem[Mare{\v{c}}ek et~al.(2016)Mare{\v{c}}ek, Shorten, and
  Yu]{marevcek2016r}
------, ``r-extreme signalling for congestion control,'' \emph{International
  Journal of Control}, vol.~89, no.~10, pp. 1972--1984, 2016.

\bibitem[Fioravanti et~al.(2017)Fioravanti, Mare\v{c}ek, Shorten, Souza, and
  Wirth]{Fioravanti17}
A.~R. Fioravanti, J.~Mare\v{c}ek, R.~N. Shorten, M.~Souza, and F.~R. Wirth,
  ``On classical control and smart cities,'' \emph{Proc. of the IEEE 56th
  Annual Conference on Decision and Control}, pp. 1413--1420, 2017.

\bibitem[Kabir and Lee(2020)]{kabir2020receding}
R.~H. Kabir and K.~Lee, ``Receding-horizon ergodic exploration planning using
  optimal transport theory,'' in \emph{2020 American Control Conference (ACC),
  IEEE. to appear. Preprint is available with DOI}, vol.~10, 2020.

\bibitem[Dwork et~al.(2012)Dwork, Hardt, Pitassi, Reingold, and
  Zemel]{dwork2012fairness}
C.~Dwork, M.~Hardt, T.~Pitassi, O.~Reingold, and R.~Zemel, ``Fairness through
  awareness,'' in \emph{Proceedings of the 3rd innovations in theoretical
  computer science conference}, 2012, pp. 214--226.

\bibitem[Boutsis and Kalogeraki(2014)]{boutsis2014task}
I.~Boutsis and V.~Kalogeraki, ``On task assignment for real-time reliable
  crowdsourcing,'' in \emph{2014 IEEE 34th International Conference on
  Distributed Computing Systems}.\hskip 1em plus 0.5em minus 0.4em\relax IEEE,
  2014, pp. 1--10.

\bibitem[Duan et~al.(2017)Duan, Li, and Cai]{duan2017distributed}
Z.~Duan, W.~Li, and Z.~Cai, ``Distributed auctions for task assignment and
  scheduling in mobile crowdsensing systems,'' in \emph{2017 IEEE 37th
  International Conference on Distributed Computing Systems (ICDCS)}.\hskip 1em
  plus 0.5em minus 0.4em\relax IEEE, 2017, pp. 635--644.

\bibitem[{Gong} et~al.(2016){Gong}, {Wei}, {Guo}, {Zhang}, and {Fang}]{7365407}
Y.~{Gong}, L.~{Wei}, Y.~{Guo}, C.~{Zhang}, and Y.~{Fang}, ``Optimal task
  recommendation for mobile crowdsourcing with privacy control,'' \emph{IEEE
  Internet of Things Journal}, vol.~3, no.~5, pp. 745--756, 2016.

\bibitem[Papadimitriou and Tsitsiklis(1999)]{Papadimitriou1999}
C.~H. Papadimitriou and J.~N. Tsitsiklis, ``The complexity of optimal queuing
  network control,'' \emph{Mathematics of Operations Research}, pp. 293--305,
  1999.

\bibitem[Whittle(1988)]{whittle1988restless}
P.~Whittle, ``Restless bandits: Activity allocation in a changing world,''
  \emph{Journal of applied probability}, pp. 287--298, 1988.

\bibitem[Weber and Weiss(1990)]{weber1990index}
R.~R. Weber and G.~Weiss, ``On an index policy for restless bandits,''
  \emph{Journal of applied probability}, pp. 637--648, 1990.

\bibitem[Ghosh et~al.(2019)Ghosh, Marecek, and Shorten]{ghosh2019iterated}
R.~Ghosh, J.~Marecek, and R.~Shorten, ``Iterated piecewise-stationary random
  functions,'' \emph{arXiv preprint arXiv:1909.10093}, 2019.

\bibitem[Barnsley(1988)]{barnsley1993}
M.~Barnsley, \emph{Fractals Everywhere}.\hskip 1em plus 0.5em minus 0.4em\relax
  Academic Press, 1988.

\bibitem[Myjak and Szarek(2003)]{myjak2003}
J.~Myjak and T.~Szarek, ``Attractors of iterated function systems and markov
  operators,'' \emph{Abstract and Applied Analysis}, no.~8, pp. 479--502, 2003.

\bibitem[Barnsley et~al.(1988)Barnsley, Demko, Elton, and
  Geronimo]{BarnsleyDemkoEltonEtAl1988}
M.~Barnsley, S.~Demko, J.~Elton, and J.~Geronimo, ``Invariant measures for
  {M}arkov processes arising from iterated function systems with
  place-dependent probabilities,'' \emph{Annales de l'institut Henri
  Poincar{\'e} (B) Probabilit{\'e}s et Statistiques}, vol.~24, no.~3, pp.
  367--394, 1988.

\bibitem[Evans(1997)]{Evans99partialdifferential}
L.~C. Evans, ``Partial differential equations and {M}onge-{K}antorovich mass
  transfer,'' \emph{Current Developments in Mathematics}, vol. 1997, no.~2, pp.
  65--126, 1997.

\bibitem[Svetlozar~T and Rüschendorf(1998)]{mass1998}
R.~Svetlozar~T and L.~Rüschendorf, \emph{Mass Transportation Problems-Vol I
  and Vol II, 1998}, ser. Probability and its Applications.\hskip 1em plus
  0.5em minus 0.4em\relax Springer, New York, NY, 1998.

\bibitem[Villani(2009)]{Villani2009OptimalT}
C.~Villani, \emph{Optimal transport: old and new}, ser. Grundlehren der
  mathematischen Wissenschaften.\hskip 1em plus 0.5em minus 0.4em\relax Berlin
  Heidelberg: Springer-Verlag, 2009, vol. 338.

\bibitem[Krajzewicz et~al.(2012)Krajzewicz, Erdmann, Behrisch, and
  Bieker]{Krajzewicz2012}
D.~Krajzewicz, J.~Erdmann, M.~Behrisch, and L.~Bieker, ``Recent development and
  applications of sumo-simulation of urban mobility,'' \emph{International
  journal on advances in systems and measurements}, vol.~5, no. 3\&4, 2012.

\bibitem[Wegener et~al.(2008)Wegener, Pi\'{o}rkowski, Raya, Hellbr\"{u}ck,
  Fischer, and Hubaux]{Wegeneretal2008}
A.~Wegener, M.~Pi\'{o}rkowski, M.~Raya, H.~Hellbr\"{u}ck, S.~Fischer, and J.-P.
  Hubaux, ``Traci: an interface for coupling road traffic and network
  simulators,'' \emph{Proc. of the 11th Communications and Networking
  Simulation Symposium, Ottawa, Canada}, pp. 155--163, 2008.

\bibitem[{City of Melbourne}(Jun. 2017. Accessed: Nov. 13,
  2019)]{CityOfMelbourneBoundaryMap}
\BIBentryALTinterwordspacing
{City of Melbourne}, ``Suburbs and postcodes: City of melbourne,'' City of
  Melbourne, GPO Box 1603, Melbourne, Victoria 3001, Australia, Jun. 2017.
  Accessed: Nov. 13, 2019. [Online]. Available:
  \url{https://www.melbourne.vic.gov.au/sitecollectiondocuments/suburb-map-boundary-city-of-melbourne.pdf}
\BIBentrySTDinterwordspacing

\bibitem[{City of Melbourne}(Sep. 2017. Accessed: Nov. 13,
  2019)]{ParkingBaysMap}
\BIBentryALTinterwordspacing
------, ``City of melbourne open data, {\it on-street parking bays},'' City of
  Melbourne, GPO Box 1603, Melbourne, Victoria 3001, Australia, Sep. 2017.
  Accessed: Nov. 13, 2019. [Online]. Available:
  \url{https://data.melbourne.vic.gov.au/Transport-Movement/On-street-Parking-Bays/crvt-b4kt}
\BIBentrySTDinterwordspacing

\bibitem[{German Aerospace Center} and others.(Accessed: July 3,
  2020)]{SUMO-striping}
\BIBentryALTinterwordspacing
{German Aerospace Center} and others., ``Simulation of urban mobility -- wiki,
  {\it simulation/pedestrians},'' Accessed: July 3, 2020. [Online]. Available:
  \url{https://sumo.dlr.de/docs/Simulation/Pedestrians.html#model_striping}
\BIBentrySTDinterwordspacing

\bibitem[Franklin et~al.(1994)Franklin, Powell, and Emami-Naeini]{Franklin}
G.~F. Franklin, J.~D. Powell, and A.~B. Emami-Naeini, \emph{{F}eedback
  {C}ontrol of {D}ynamic {S}ystems}.\hskip 1em plus 0.5em minus 0.4em\relax
  Reading, MA: Addison Wesley, 1994.

\bibitem[Franklin et~al.(1997)Franklin, Powell, and Workman]{Franklin_dig}
G.~F. Franklin, J.~D. Powell, and M.~L. Workman, \emph{{D}igital {C}ontrol of
  {D}ynamic {S}ystems}, 3rd~ed.\hskip 1em plus 0.5em minus 0.4em\relax
  Englewood Cliffs, NJ: Prentice Hall, 1997.

\bibitem[Armstrong()]{RFIDInsider}
\BIBentryALTinterwordspacing
S.~Armstrong, ``6 factors that affect rfid read range,'' \emph{RFID Insider}.
  [Online]. Available:
  \url{http://blog.atlasrfidstore.com/improve-rfid-read-range}
\BIBentrySTDinterwordspacing

\bibitem[Narasimhan et~al.(2020)Narasimhan, Cotter, Gupta, and
  Wang]{narasimhan2020pairwise}
H.~Narasimhan, A.~Cotter, M.~R. Gupta, and S.~Wang, ``Pairwise fairness for
  ranking and regression,'' in \emph{The Thirty-Fourth AAAI Conference on
  Artificial Intelligence (AAAI-20)}, 2020, pp. 5248--5255.

\bibitem[Zhou et~al.(2021)Zhou, Marecek, and Shorten]{zhou2020fairness}
Q.~Zhou, J.~Marecek, and R.~N. Shorten, ``Fairness in forecasting and learning
  linear dynamical systems,'' in \emph{The Thirty-Fifth AAAI Conference on
  Artificial Intelligence (AAAI-21)}, 2021, arXiv preprint arXiv:2006.07315.

\bibitem[Marecek(2020)]{marecek2020screening}
J.~Marecek, ``Screening for an infectious disease as a problem in stochastic
  control,'' \emph{arXiv preprint arXiv:2011.00635}, 2020.

\bibitem[Bogachev(2007)]{bogachev-2007}
V.~I. Bogachev, \emph{Measure theory Vol 2}.\hskip 1em plus 0.5em minus
  0.4em\relax Berlin Heidelberg: Springer-Verlag, 2007.

\bibitem[Prokhorov(1956)]{Prokhorov1956ConvergenceOR}
Y.~V. Prokhorov, ``Convergence of random processes and limit theorems in
  probability theory,'' \emph{Theory of Probability and Its Applications},
  vol.~1, pp. 157--214, 1956.

\bibitem[Kryloff and Bogoliouboff(1937)]{Kryloff37}
N.~Kryloff and N.~Bogoliouboff, ``La théorie g{\' e}n{\' e}rale de la mesure
  dans son application {\` a} l'{\' e}tude des syst{\` e}mes dynamique de la
  m{\' e}canique non lin{\' e}aire,'' \emph{Ann. of Math. II}, vol.~38, pp.
  65--113, 1937.

\bibitem[Hairer et~al.(2011)Hairer, Mattingly, and
  Scheutzow]{hairer2011asymptotic}
M.~Hairer, J.~C. Mattingly, and M.~Scheutzow, ``Asymptotic coupling and a
  general form of harris’ theorem with applications to stochastic delay
  equations,'' \emph{Probability theory and related fields}, vol. 149, no. 1-2,
  pp. 223--259, 2011.

\bibitem[Elton(1987)]{Elton1987ergodic}
J.~H. Elton, ``An ergodic theorem for iterated maps,'' \emph{Ergodic Theory and
  Dynamical Systems}, vol.~7, no.~04, pp. 481--488, 1987.

\end{thebibliography}
\appendix 
\normalsize 
\subsection{An Alternative Formalisation}
\label{sec:A}
While the formalisation of Section III of the paper is perfectly valid, one could also consider an alternative formalisation. 
Therein, consider $N$ agents aiming to estimate the evolution of state transitions (for finite state space; or an evolution of a measure for uncountable state space) of a state-space Markov chain $\{X_k\}_{k\ge 0}$. Let $\mathcal X=\{1,2,\dots, n\}$ when the underlying state-space is finite and let $\mathcal X$ be a closed subset of $\mathbb R^n$ with a metric $\rho$ on it such that $(\mathcal X, \rho)$ forms a complete, separable metric space when the underlying state-space is uncountable. Let $\nu_0$ be the initial distribution of the Markov chain. Each agent plays its role once in a predetermined (but in a random fashion) sequential order indexed by $k=0,1,2,\dots$ which can be viewed as a discrete time instant. Then:

\begin{dfn}[Social sensing]\label{dfn:socialsensing}
The social sensing protocol by a group of $N$ agents proceeds as follows: 
\begin{itemize}
\item Initially, let the chain start from some $\mathcal{A}\in \mathcal B(\mathcal X)$ i.e., $\mathbb P (X_0\in \mathcal{A})=\nu_0(\mathcal{A})$.
\item  At time $k$, $\mathcal C$ broadcasts a signal $\pi(k)$ and the agents change their state state from $X_{k-1}$ to $X_k$ such that 
\begin{align*}
\mathbb P(X_k\in \mathcal A| X_{k-1})&=\int_{\mathcal X} \nu(X_{k-1}, \mathcal A) d\nu_{k-1}(x_{k-1})\\
&=\nu_k(\mathcal A)
\end{align*}
for any event $\mathcal A\in \mathcal B(\mathcal X)$, and the integral is in the sense of Lebesgue with respect of probability measure $\nu_{k-1}$ and $\nu(x,\mathcal A)= \mathbb P(X_k \in \mathcal A| X_{k-1}=x)$ is the transition kernel. 
\end{itemize}
\end{dfn}
Now at this point we should mention that since there is an inherent randomness in the reaction of each agent to the broadcast signal, the closed loop of Figure \ref{system} requires a stochastic model. A model based on an iterated function system (IFS), which is a class of discrete-time Markov Chains with an uncountable state space, considers response functions that are either absolutely continuous or Lipschitz in nature. 
If for $N$ agents, we use $f_1,\dots, f_N$ ``response functions'', then the chain will  move by the action of a single randomly chosen function out of the family. This can be stated in the following recurrent relation:
\begin{align*}
X_{k+1}=f_{\sigma_k}(X_k) \quad, k=0,1,2,\dots, 
\end{align*}
$\sigma_0, \sigma_1,\dots$ are i.i.d discrete random variables taking values in $[1,N]$.
Now, for the evolution of this social sensing process we define:
\begin{dfn}[A linear operator for social sensing]\label{dfn: lin-op-fr-so-sens}
Let us consider an operator $P$ on the space of all bounded continuous functions $C_b(\mathcal X)$, for social sensing as follows:
\begin{align*}
Pf(x)= \mathbb E \big[f\left(X_{k+1}\right)| X_k=x\big]
\end{align*}
whose dual $P^{\star}$ is defined on the space of all Borel probability measure  on $\mathcal X$ denoted by $\mathcal M(\mathcal X)$ as follows:
\begin{align*}
P^{\star}\nu=\int_{\mathcal X} \mathbb P(X_k\in \mathcal A| X_{k-1}=x)\nu dx    
\end{align*}
\end{dfn}
As  stated in the introduction,  our aim  in  this paper is to regulate the task distribution of social sensing, which assures predictability and fairness. An important prerequisite for both predictability and fairness is:

\begin{dfn}[Ergodicity]\label{dfn:ergo}
Let us consider a linear operator $P$ for social sensing.
We call the social-sensing ergodic if there exists an unique $\nu_{\star}\in \mathcal M(\mathcal X)$ such that $P_{\star}\nu^{\star}=\nu_{\star}$.
\end{dfn}
which could be studied independently. 

\subsection{Proof of the Main Result}\label{sec:B}
In the following, we present the proof of the main result.
\begin{proof}[Proof of Theorem \ref{thm:purturbation}]
Letting $C_b(\mathcal X)$ denote the set of real-valued bounded continuous functions on $\mathcal X$, one can define a linear map  $P$ on $C_b(\mathcal X,\mathbb R)$ (using Definition~\ref{def:markov}):
\begin{align}\label{eq:P-action}
Pw(x):=\sum_{i=1}^{N} p_{i}(x)(w\circ w_{\sigma_{i}})(x) 
\end{align}
This operator characterizes the Markov chain. $P$ maps $C_b(\mathcal X)$ into itself, which is known as Feller property, or that $P$ is a Feller map.
Let $\mathcal M(\mathcal X)$ denote the set of Borel probability measures on $\mathcal X$. Denote the dual of the map $P$ as follows:
\begin{align}\label{eq:dual-1}
P^{\star}: \mathcal M(\mathcal X)\to \mathcal M(\mathcal X),
\end{align}
with the requirement 
\begin{align}\label{eq:dual-requ}
\int_{\mathcal X} w d (P^{\star} \nu)=\int_{\mathcal X} (P w) d \nu.
\end{align}
Such a dual map $P^{\star}$ is well defined by the Riesz representation theorem.
Now we show that $P^{\star}$ is a contraction in Wasserstein-$1$, i.e., in $\mathcal W_1$ metric with some contraction factor $r\in (0,1)$. For any two $\nu_1, \nu_2\in \mathcal M(\mathcal X)$, we have:
\begin{align}\label{eq:wass-contraction}
&\mathcal W_1(P^{\star}\nu_1, P^{\star} \nu_2)\nonumber\\
&=\sup\limits_{w\in \mathcal L_1}\left [ \int w d(P^{\star}\nu_1)- \int wd(P^{\star}\nu_2)\right ]\nonumber\quad \left(\because \eqref{eq:wass-dist}\right)\\
&=\sup\limits_{w\in \mathcal L_1}\left [ \int (Pw)d\nu_1- \int (Pw)d\nu_2\right ]\nonumber\quad\left(\because \eqref{eq:dual-requ}\right)\\
&=\sup\limits_{w\in \mathcal L_1}\left [ \int (Pw)d(\nu_1-\nu_2)\right ]\nonumber\\
&=r\cdot\sup\limits_{w\in \mathcal L_1}\left [ \int \left(\frac{1}{r}Pw\right)d(\nu_1-\nu_2)\right ]\nonumber\\
&= r\cdot\sup\limits_{g\in \mathcal L_1} \int gd(\nu_1-\nu_2)\quad\left (\because g:=\frac{1}{r}Pw\in \mathcal L_1 \text{ as } w\in \mathcal L_1\right )\nonumber\\
&\le r\cdot\mathcal W_1(\nu_1, \nu_2).
\end{align}
Now a useful consequence of the above derived fact is: 
\begin{align}\label{eq:esti-2}
&\mathcal W_1(\nu_1^{\star}, \nu_{2}^{\star})\nonumber \\
&= \mathcal W_1(P_1^{\star}\nu_1^{\star}, P_{2}^{\star}\nu_{2}^{\star})\nonumber \quad \left(\because P_1^{\star}\nu_1^{\star}=\nu_1^{\star}, P_2^{\star}\nu_2^{\star}=P_2^{\star}\nu_2^{\star} \right)\\
& \le \mathcal W_1(P_1^{\star}\nu_1^{\star}, P_1^{\star}\nu_{2}^{\star})+\mathcal W_1(P_1^{\star}\nu_{2}^{\star}, P_{2}^{\star}\nu_{2}^{\star})\nonumber \left(\text{Triangle inequality}\right)\\
& \le  r\mathcal W_1(\nu_1^{\star},\nu_{2}^{\star})+\mathcal W_1(P_{1}\nu_{2}^{\star}, P_{2}\nu_{2}^{\star})\nonumber \quad \left( \because \eqref{eq:wass-contraction}\right)\\
& \Rightarrow \mathcal W_1(\nu_1^{\star}, \nu_{2}^{\star}) \le \frac{(\mathcal W_1P_{1}\nu_{2}^{\star}, P_{2}\nu_{2}^{\star})}{1-r}
\end{align}
Now, notice that:
\begin{align}\label{eq:esti-3}
&\Big\|{P_1w(x)-P_{2}w(x)} \Big\|\nonumber\\
& \overset{\eqref{eq:P-action}}{=}\Big\|\sum\limits_{\sigma_{k}}p_{\sigma_{k}}(x)(w\circ  w_{\sigma_{k}})(x)-\sum\limits_{\sigma_k}p'_{\sigma_{k}}(x)(w\circ w'_{\sigma_{k}})(x)\Big\|\nonumber\\
&=\Big\|\sum\limits_{\sigma_{k}}p_{\sigma_{k}}(x)\left[(w\circ  w_{\sigma_{k}})(x)-(w\circ w'_{\sigma_k})(x)+(w\circ w'_{\sigma_k})(x)\right]\nonumber\\
& -\sum\limits_{\sigma_k}p'_{\sigma_{k}}(x)(w\circ w'_{\sigma_{k}})(x)\Big\|\nonumber\\
&=\Big\|\sum\limits_{\sigma_{k}}p_{\sigma_{k}}(x)\left[(w\circ  w_{\sigma_{k}})(x)-(w\circ w'_{\sigma_k})(x)\right]+ \nonumber\\
& \sum\limits_{\sigma_{k}}p_{\sigma_{k}}(x)(w\circ w'_{\sigma_k})(x)-\sum\limits_{\sigma_k}p'_{\sigma_{k}}(x)(w\circ w'_{\sigma_{k}})(x)\Big\|\nonumber\\
&=\Big\|\sum\limits_{\sigma_{k}}p_{\sigma_{k}}(x)\left[(w\circ  w_{\sigma_{k}})(x)-(w\circ w'_{\sigma_k})(x)\right]+ \nonumber\\
& \sum\limits_{\sigma_{k}}\left[p_{\sigma_{k}}(x)-p'_{\sigma_{k}}(x)\right](w\circ w'_{\sigma_k})(x)\Big\|\nonumber\\
&\le \Big\| {\sum\limits_{\sigma_{k}}p_{\sigma_{k}}(x)\left[(w\circ w_{\sigma_{k}})(x) -(w\circ w'_{\sigma_{k}})(x)\right]}\Big\|+\nonumber\\
& \Big\|\sum\limits_{\sigma_{k}}\left[p_{\sigma_{k}}(x)-p'_{\sigma_{k}}(x)\right](w\circ w'_{\sigma_{k}})(x)\Big\|\nonumber\\
&\le\sum\limits_{\sigma_{k}} \Big\|p_{\sigma_{k}}(x)\Big\|\cdot \Big\|\left[(w\circ w_{\sigma_{k}})(x) -(w\circ w'_{\sigma_{k}})(x)\right]\Big\|+\nonumber\\
&\sum\limits_{\sigma_{k}}\Big\|\left[p_{\sigma_{k}}(x)-p'_{\sigma_{k}}(x)\right]\Big\|\cdot \Big\|(w\circ w'_{\sigma_{k}})(x)\Big\|\nonumber\\
&\le r'\sum\limits_{\sigma_{k}} p_{\sigma_{k}}(x)\Big\|w_{\sigma_{k}}(x) - w'_{\sigma_{k}}(x)\Big\|+\beta \eta.
\end{align}

And, then,
\begin{align*}
&\mathcal W_1(P_{1}^{\star}\nu,P_2^{\star}\nu) \nonumber\\
=&\sup\limits_{w}\int w d(P_1^{\star}\nu-P_{2}^{\star}\nu)\nonumber\\
=&\sup\limits_{w}\int (P_1w-P_{2}w)d\nu\nonumber\\
\le&\sup\limits_{x}\left (r'\sum\limits_{\sigma_{k}} p_{\sigma_{k}}(x)\Big\|w_{\sigma_{k}}(x) - w'_{\sigma_{k}}(x)\Big\|+\beta \eta\right)\nonumber\\
\le &\left (r'\sum\limits_{\sigma_k} p_{\sigma_{k}}(x)\Big\|{w_{\sigma_{k}}(x) - w'_{\sigma_{k}}(x)}\Big\|_{\infty}+\beta \eta\right).
\end{align*}
And, finally \eqref{eq:esti-1} is concluded from \eqref{eq:esti-2} and \eqref{eq:esti-3}.
\end{proof}

\section{Analysis in the time-varying case}\label{sec:C}
Next, we would like to show the existence of a certain family of measures (Theorem \ref{thm:existence})
and its uniqueness (Theorem \ref{thm:coupling}) in the time-varying case. In Theorem \ref{thm:existence}, we need:
\begin{dfn}[Uniformly tight measure; Definition 8.6.1 in Bogachev \cite{bogachev-2007}]\label{dfn:uni-tight}
An arbitrary $\mathcal M\subseteq \mathcal M(\mathcal X)$ is called uniformly tight if $\forall \epsilon >0$
there exists a compact subset $\mathcal K\subseteq \mathcal X$ such that $\nu(\mathcal K)\ge 1-\epsilon,\quad \forall \nu\in \mathcal M(\mathcal X)$.
\end{dfn}
It can be shown that on a compact metric space, any family of probability measures is uniformly tight, cf. Theorem 8.6.2 in \cite{bogachev-2007}.
Intuitively, for any other space, probability measures accumulate on compact subsets of the underlying space. We use a result due to Prokhorov \cite{Prokhorov1956ConvergenceOR} which says, if $\{\nu_n\}_{n=1}^{\infty} \in \mathcal M(\mathcal X)$ be uniformly tight sequence, then there exists a sub-sequence $\{\nu_{n_k}\}_{k=1}^{\infty}$ of $\{\nu_n\}_{n=1}^{\infty}$ and a $\nu\in \mathcal M(\mathcal X)$ such that $\{\nu_{n_k}\}_{n_k}\to \nu$ weakly. Now, with this in mind, we establish the existence of invariant measures of the Markov process described in equation  \eqref{eq:tvirf}.
Let $\mathcal B(\mathcal X)$ denote the Borel sigma-algebra on $\mathcal X$. For any Borel set $\mathcal A\in \mathcal B (\mathcal X)$ let us define $m$-step transitional probability functions, which are probability measure for each fixed $x\in \mathcal X$ and measurable function of $x$ for each fixed $\mathcal A\in \mathcal{B}(\mathcal X)$, as follows:
\begin{align}\label{eq:m-step-transi}
\nu^{s}_{m}\left(x, \mathcal A\right)=\text{ Prob }\left (x^s(m)\in \mathcal A| x^s(0)=x\right).
\end{align}
\begin{proof}[Proof of Theorem \ref{thm:existence}]
Assume that there exists at least one $x\in \mathcal X$ for which the sequence $\{\nu^s_j(x, \cdot)\}_{j=0}^{\infty}$ is uniformly tight. Then we show that there exists at least one invariant probability measure for $P_s^{\star}$. The proof is based on the Krylov-Bogoliubov \cite{Kryloff37} type argument. 
Define a sequence of probability measures which are the average over time of the $m$-step transition probabilities on $(\mathcal X, \mathcal B (\mathcal X))$ as follows for some fixed  $x\in \mathcal X$: 
\begin{align}\label{eq:avrg-trnsi-pro}
\text{ for } \mathcal A\in \mathcal B (\mathcal X),\quad 
\nu^s_m(\mathcal A)=\frac{1}{m}\sum\limits_{j=1}^{m}\nu^s_j (x,\mathcal A)
\end{align}
It is clear that this sequence is also tight, so it has a sub-sequence that converges weakly to some probability measure $\nu^s_{\star}$ on $\mathcal X$. We also have the following equality:
\begin{align}
P_s^{\star}\nu^{s}_m- \nu_m^{s}&= \frac{1}{m}\sum\limits_{j=2}^{m+1}\nu^s_j (x,\mathcal A)-\frac{1}{m}\sum\limits_{j=1}^{m}\nu^s_j (x,\mathcal A)\nonumber\\
&=\frac{1}{m}\left [\nu^s_{m+1} (x,\mathcal A)-\nu^s_1 (x,\mathcal A) \right]
\end{align}
Notice that for each fixed $x\in \mathcal X$, $\nu^s(x,\mathcal A)$ is a probability measure and the integral of $w(x)$ with respect to such measure is expressed as $\int w(y)\nu^s(x,dy)$, and the interpretation holds for any $m\in \mathbb N$ and written as $\int w(y)\nu^s_{m}(x,dy)$.
Take any $w\in C_b(\mathcal X, \mathbb R)$ such that $|w(x)|<1$. Fix an $\epsilon >0$. Weak convergence of the probability measures $\{\nu^s_m\}_{m\ge 1}$ ensures that there is a natural number $m>\frac{1}{\epsilon}$ for which
\begin{equation*}
\left |\int w(x) \nu^s_m(dx)-\int w(x) \nu^s_{\star}(dx)\right|\le \epsilon.   
\end{equation*}
Since $(P_sw)$ is continuous, we can chose large $m$ for which 
\begin{equation*}
\left |\int (P_sw)(x) \nu^s_m(dx)-\int (P_sw)(x) \nu^s_{\star}(dx)\right|\le \epsilon. \end{equation*}
Now, 
\begin{align*}
&\left|\int w(x) (P_s^{\star}\nu^s_{\star})(dx)-\int w(x)\nu^s_{\star}(dx)\right|\\
&\le \left |\int w(x) (P_s^{\star}\nu^s_{\star})(dx)-\int w(x)(P_s^{\star}\nu^s_m)(dx)\right|\\
&+\left |\int w(x) (P_s^{\star}\nu^s_m)(dx)-\int w(x)\nu^s_m(dx)\right|\\
&+ \left|\int w(x)\nu^s_m(dx)-\int w(x)\nu^s_{\star}(dx)\right|\\
&\le \left|\int (P_s w)(x)\nu^s_{\star}(dx)-\int (P_s w)(x)\nu^s_m(dx)\right|\\
&+ \frac{1}{m}\left |\int w(y)(\nu^s_{m+1}(x,dy)-\int w(y)\nu_s(x,dy)\right|+\epsilon\\
&\le  2\epsilon+ \frac{2}{m}\le 4\epsilon
\end{align*}
Since the above relation is true for any arbitrary $\epsilon$, we can conclude
\begin{equation*}
\left|\int w(x) (P_s^{\star}\nu^s_{\star})(dx)-\int w(x)\nu^s_{\star}(dx)\right|=0
\end{equation*}
Also, considering that $w\in C_b(\mathcal X, \mathbb R)$ is arbitrary, 
\begin{align*}
P_s^{\star}\nu^s_{\star}=\nu^s_{\star}.    
\end{align*}
\end{proof}
Next, notice that any two trajectories get arbitrarily close to each other, eventually:
\begin{thm}\label{thm:coupling}
Consider two trajectories (realizations) of the Markov chain in \eqref{eq:tvirf} starting from any two different initial conditions $x^s(0)$ and $y^s(0)$. These trajectories couple in the sense of (1.2) in Hairer \cite{hairer2011asymptotic}.
\end{thm}
\begin{proof}[Proof of Theorem \ref{thm:coupling}]
Consider the two trajectories of the system \eqref{eq:tvirf}
with $w_{\sigma_k^s}(x^s(k))=A_{\sigma_k^s}x^s(k)+b_{\sigma_k^s}$
starting from two different initial condition $x^s(0)$ and $y^s(0)$ as follows, where $s$ is denote the discrete-time scale over $\mathbb N$:
\begin{align}
&x^s(k)=\left(w_{\sigma^s_{k-1}}\circ w_{\sigma^s_{k-2}}\circ\cdots\circ w_{\sigma^s_1}\right)(x^s(0))\\
&y^s(k)=\left(w_{\sigma^s_{k-1}}\circ w_{\sigma^s_{k-2}}\circ\cdots\circ w_{\sigma^s_1}\right)(y^s(0))
\end{align}
Let $\norm{\cdot}$ be any norm on $\mathbb R^n$, then any $n\times n$ real-matrix $A$ induces a linear operator on $\mathbb R^n$ with respect to the standard basis and norm of $A$ is well defined as 
\begin{align}
\norm{A}:= \sup\limits_{x\ne 0}\Bigg\{\frac{\norm{Ax}}{\norm{x}}: x\in \mathbb R^n\Bigg\}
\end{align}
Since all matrices involved in the transformations are Schur matrices (i.e., if $\lambda$ is an eigenvalue for such matrix, $|\lambda|<1$) then for any matrix norms induced by vector norms $\norm{\cdot}$, we have the following:
\begin{align}\label{lessone}
0<\norm{\prod\limits_{i=1}^{k-1} A_{\sigma_i^s}} \le  \prod\limits_{i=1}^{k-1}\norm{A_{\sigma_i^s}}< \prod\limits_{i=1}^{k-1}\lambda_{\sigma_i^s}< \left(\hat \lambda \right)^k<1,
\end{align}
where $\lambda_{\sigma_i^s}<1$ is the largest eigenvalue of the matrix $A_{\sigma_i^s}$ and $\hat \lambda$ is the largest of all such $\{\lambda_{\sigma_i^s}\}_{i=1}^{k-1}$. 
One can notice that for all initial values $x^s(0), y^s(0)\in \mathcal X$ we have,
\begin{align*}
\rho(x^s(k), y^s(k))&=\norm{x^s(k)- y^s(k)}\\
&= \left(\norm{\prod\limits_{i=1}^{k-1} A_{\sigma_i^s}}\right) \norm{x^s(0)-y^s(0)}\\
&\le \left (\prod\limits_{i=1}^{k-1}\norm{A_{\sigma_i^s}}\right) \rho(x^s(0), y^s(0))\\
&\le \left(\hat \lambda \right)^k \rho(x^s(0), y^s(0)) \xrightarrow{\mathit{k\to \infty}} 0. \because \eqref{lessone}.
\end{align*}
Thus, the trajectories couple as $k\to \infty$. 
\end{proof}

\subsection{Implications for the time-varying case}\label{sec:D}
Our results of Appendix \ref{sec:C} assure predictability, as introduced in Definition~\ref{dfn:pred}, even in the time-varying case:
\begin{cor}[Predictability in the Time-Varying Setting] \label{cor:pred}
Consider the feedback system depicted in Fig. \ref{system}, for some given finite-dimensional stable linear systems ${\mathcal C}$ and ${\mathcal F}$. Assume that each agent $i \in \{1,\dots,N\}$ has state $x_i^s(k)$ governed by  \eqref{eq:tvirf}, where the affine mapping $w_{ij}^s$ is chosen at each step of time according to a Dini-continuous
probability function $p_{ij}^s(x_i^s(k),\pi(k))$, out of
\eqref{eq:func-2}, where $A_i$ is a Schur matrix and for all $i$, $\pi(k)$, and for each fixed $s$, $\sum_j p^s_{ij}(x^s_i(k),\pi(k)) = 1$. Moreover, assume that the probabilities $p_{ij}^s$ are bounded away from zero and that the conditions of Theorem \ref{thm:existence} hold for the time-varying process thus defined. Then, the feedback loop ensures {\em predictability} to each agent's dynamics, i.e., for each agent $i$, there exists a constant $\overline{r}_i$ such that 
\begin{equation}
\lim_{k\to \infty} \frac{1}{k+1} \sum_{j=0}^k x^s_i(j) = \overline{r}_i \quad \textrm{a.s}.
\end{equation}
\end{cor}
\begin{proof}[Proof of Corollary \ref{cor:pred}]
Predictability follows from the existence of an ergodic measure (Theorem \ref{thm:existence}), its uniqueness (Theorem \ref{thm:coupling}), and from
Theorem 2 of Elton \cite{Elton1987ergodic}. 
\end{proof}
In turn, predictability allows for fairness, as introduced in Definition~\ref{dfn:fair}, albeit under strict conditions suggesting that the agent's behaviour is symmetric, in some sense, and their initial states are the same:
\begin{cor}[Fairness in the Time-Varying Setting] \label{cor:fair2}
Consider the feedback system depicted in Fig. \ref{system}, and the same conditions as in Corollary \ref{cor:pred} in the time-varying case. 
If, in addition:
\begin{itemize}
\item the agents' states evolve from a uniform initial state, that is, in the time-varying case, 
if there exists a constant $c$ such that $x_i^0(0) = c$ for all agents $i$, 
\item the probability tuple \eqref{condprob} is uniform, i.e., there exists a family of functions $\{ q^s_\sigma \}_{s=1}^{\infty}$, $q^s_\sigma: \mathcal X\to [0,1]$, such that for all agents $i=1,2 \ldots N$ and all times $s$, $p_i^s(x)=q^s(x)$,
\end{itemize}
then the feedback loop ensures {\em fairness} of the agents' dynamics within each segment $s$. That is, for all segments $s$, there exists a constant $\overline{r}^s$ such that for all agents $i$ 
\begin{equation}
\label{eq:time-varying-fair}
\lim_{k\to \infty} \frac{1}{k+1} \sum_{j=0}^k x^s_i(j) = \overline{r}^s \quad \textrm{a.s}.
\end{equation}
\end{cor}
\begin{proof}[Proof of Corollary \ref{cor:fair2}]
Fairness follows from the Markov property of the time-varying model \eqref{eq:tvirf}.
\end{proof}
As above:
\begin{rem}
Throughout Corollaries
\ref{cor:pred}--\ref{cor:fair2}, Dini's condition on the probabilities may be replaced by simpler, more conservative assumptions, such as Lipschitz or H\"older continuity \cite{BarnsleyDemkoEltonEtAl1988}.
\end{rem}
The results presented in Theorem \ref{thm:coupling} and in Corollary \ref{cor:pred} ensure that the participants' trajectories still couple for different initial conditions; that is, predictability still holds.
\clearpage
\onecolumn
\appendix
\section{An Overview of Notation}
\label{app:symbols}
\begin{tabularx}{\linewidth}{l|X}
\caption{A Table of Notation}\\\toprule\endfirsthead
\toprule\endhead
\midrule\multicolumn{2}{l}{\itshape continues on next page}\\\midrule\endfoot
\bottomrule\endlastfoot
\textbf{Symbol} & \textbf{Meaning} \\
\midrule
$\mathbb{N}$                               & the set of natural numbers.\\
$\mathbb{Q}$                               & the set of rational numbers.\\
$\mathbb{R}$                               & the set of real numbers.\\
${\mathbb Z}$                              & set of all integers.\\
$\mathcal A$                               & a Borel set i.e an event i.e a typical element in $\mathcal B(\mathcal X)$, Definition~\ref{dfn:ifs}.\\
$\mathcal A^c$                             & the complement of the event or the Borel set $\mathcal A$, \eqref{eq:char-func}.\\
${\mathcal{C}}$                            & controller representing the central authority.\\
${\mathcal{F}}$                            & a filter.\\
$\mathcal K$                               & a compact subset in $\mathcal X$.\\
$\mathcal M$                               & a subset in $\mathcal M\left(\mathcal X\right)$.\\
$\mathcal X$                               & a metric space.\\ 
$\mathcal{L}_1$                            & the space of all Lipschitz maps with Lipschitz constant $1$.\\
$1_{\mathcal A}$                           & the characteristic function of $\mathcal A$.\\
$\mathcal B(\mathcal X)$                   & Borel sigma algebra on $\mathcal X$.\\
$\mathcal M (\mathcal X)$                  & a space of all probability measure on $\mathcal X$.\\
$C_b(\mathcal X, \mathbb R)$               & Banach space of real-valued continuous functions from $\mathcal X$ to $\mathbb R$.\\
$\mathcal W_1(\nu_1, \nu_2)$               & a Wasserstein-$1$ distance between two probability measure $\nu_1$ and $\nu_2$.\\
$\mathcal{S}_1,\dots, \mathcal{S}_N$       & $N$ agents in the network.\\
$A_i$                                      & agent's state transformation matrix, which is assumed to be Schur.\\
$A^{'}_i$                                   & after purturbation, agents state transformation matrix which is assumed Schur matrix.\\
$A_{\sigma_k^s}$                            & at $k^{\text{th}}$ time step, state-transformation matrix for a randomly selected $\sigma_{k^s}$, when  when $s^{\text{th}}$ probability-tuple-function is active.\\
$P$                                         & a Markov operator.\\
$P^{\star}$                                 & a dual of the Markov operator $P$.\\
$P_1^{\star}$                               & a Markov operator, cf. Theorem \ref{thm:purturbation}.\\
$P_2^{\star}$                               & a Markov operator, cf. Theorem \ref{thm:purturbation}.\\
$P_s$                                       & a Markov operator  when $s^{\text{th}}$ probability-tuple-function is active.\\
$P^{\star}_s$                               & a dual of the Markov operator $P_s$  when $s^{\text{th}}$ probability-tuple-function is active.\\
$N$                                         & the total number of agent in the network.\\
$X_k$                                       & a discrete-time-homogeneous Markov chain with state-space $\mathcal X$.\\
$\alpha$                                    & a constant \eqref{eq:simpl-cntrlr}.\\
$\beta$                                     & a bound for the real-valued continuous functions, cf. Theorem \ref{thm:purturbation}.\\
$\gamma$                                    & a constant \eqref{eq:simpl-cntrlr}.\\
$\delta$                                    & a bound on the rate of time steps between perturbations.\\
$\epsilon>0$                                & a sufficiently small strictly positive real number.\\
$\eta$                                      & a bound on the perturbation in probabilities, cf. Theorem \ref{thm:purturbation}.\\
$\kappa$                                    & a constant \eqref{eq:simpl-cntrlr}.\\
$\rho$                                      & a metric on the metric space $\mathcal X$.\\
$\delta_i$                                  & a lower bounds for $p_{ij}$.\\
$\pi(k)$                                    & the signal broadcast at time $k$.\\
$\nu_1^{\star}$                             & a unique invariant probability measure for the Markov operator $P_1^{\star}$, cf. Theorem \ref{thm:purturbation}.\\
$\nu_2^{\star}$                             & a unique invariant probability measures for the Markov operator $P_2^{\star}$, cf. Theorem \ref{thm:purturbation}.\\
$\nu_k$& distribution of $x_k\in \mathbb R$.\\
$\nu(x, \mathcal A)$                        & a transitional probability measure, i.e., probability of transition from a point $x$ to some $\mathcal A\in B(\mathcal X) $.\\
$\sigma_0,\sigma_1,\dots$                   & i.i.d. discrete random variables taking values in $[1,N]$.\\
$\lambda_{\sigma_i^s}$                      & the largest eigenvalues of the matrix $A_{\sigma_i^s}$\\
$\hat \lambda$                              & largest element in the set of  $\{\lambda_{\sigma_i^s}\}_{i=1}^{k-1}$.\\
$i$                                         & an agent in the network.\\
$\hat p$                                    & a predictability vector, \eqref{eq:pred-vec}.\\
$\hat f$                                    & a fairness vector, \eqref{eq:fair-vec}.\\
$r$                                         & a typical constant lies in the open interval $(0,1)\subseteq \mathbb R$, cf. Theorem \ref{thm:purturbation}.\\
$s\in \mathbb N$                            & notation for a discrete-time scale, \ref{sub-sec:time-varying}.\\ 
$\overline{r}_i$                            & the almost sure limit of predictability for agent $i$.\\
$\overline{r}$                              & the almost sure limit of fairness for agent $i$.\\
$r'$                                        & a typical real-constant used in the Theorem \ref{thm:purturbation}.\\
$x_i(k)$                                    & the usage of the resource by the $i^{\text{th}}$ agent at the time instant $k$.\\
$y(k)$                                      & a sum of all usage $x_i(k)$ of the resource by the all $N$ agents at the time instant $k$.\\
$\hat{y}$& an estimate of $y$.\\
$p_{ij}\left(x_i(k),\pi(k)\right)$          & Dini continuous probability functions of agents for agent $i$.\\
$b_{ij}$                                    & a constant term whenever $w_{ij}$ is an affine mapping\\
$p_{i1}$                                    & the probability that the $i^{\text{th}}$ agent is at state $1$ or on, at the time instant $k$, \eqref{eq:state-dep-prb1}.\\
$p_{i0}$                                    & the probability that the $i^{\text{th}}$ agent is at state $0$ or off at the time instant $k$, \eqref{eq:state-dep-prb2}.\\
$\{w_{\sigma}(x)\}_{\sigma=1}^{N}$          & a family of response functions for agents.\\
$\{p_{\sigma}(x)\}_{\sigma=1}^{N}$          & a family of probability functions of agents choosing response functions.\\
$\mathbf{p}^i(x^i(k))=(p_1^i(x^i(k)),\dots, p_N^i(x^i(k)))$& a probability tuple.\\
$e(k)$                                      & the error signal.\\
$\mathbf{1}^{\top}=\left(1,1,\dots, 1\right)^{\top}$                 & vector of $1$.\\
$\{p_i(x)\}_{i=1}^{N}$                      & probability functions used in Theorem \ref{thm:purturbation}.\\
$\{p'_i(x)\}_{i=1}^{N}$                     & perturbed probability functions, cf. Theorem \ref{thm:purturbation}.\\
$\{\mathbf{p}^s(x)=(p_1^s(x), p_2^s(x),\dots, p_N^s(x))\}_{s=1}^{\infty}$&a countable family of $N$-tuple probability functions.\\
$\sigma^s_0,\sigma^s_1,\dots$               & i.i.d discrete-random-variable taking values in $\{1,2,\dots, N\}$ when $s^{\text{th}}$ probability-tuple-function is active.\\
$x^s(k)$                                    & state at time-instant $k$ and when $s^{\text{th}}$ probability-tuple-function is active.\\
$x_i^s(k)$                                  & state of $i^{\text{th}}$ agent at time-instant $k$ and when $s^{\text{th}}$ probability-tuple-function is active.\\
$x^i(0)$                                    & initial state for the agent $i$.\\
$\{w_{\tau}\}$                              & a set of valid possible behaviors for any agent.\\
$w_{ij}$                                    & an affine map in Theorem \ref{thm:theorem-one}.\\
$[0,1]$                                     & unit closed interval in $\mathbb R$ i.e all $x\in \mathbb R$ such that $0\le x\le 1$.\\
$[1,N]:=\{1,2,\dots,N\}$                    & a finite set consists of number of agents in the network.\\
\end{tabularx}
\end{document}